\documentclass{llncs}
\usepackage{amssymb}
\usepackage{amsmath}
\usepackage{latexsym}
\usepackage{textcomp}
\usepackage{epsfig}
\usepackage{subfigure}
\usepackage{xspace}
\usepackage{paralist} 
\usepackage{thm-restate} 
\usepackage{lineno}

\newcommand{\nt}{{NodeTrix}\xspace}
\DeclareMathOperator{\skel}{skel}

\newcommand{\remove}[1]{}

\renewcommand{\qed}{\hfill $\square$}

\begin{document}
	\title{NodeTrix Planarity Testing with Small Clusters}
	
	
	\author{
		Emilio Di Giacomo\inst{1},
		Giuseppe Liotta\inst{1},\\
		Maurizio Patrignani\inst{2},
		Ignaz Rutter\inst{3},
		Alessandra Tappini\inst{1}
	}
	
	\date{}
	
	\institute{
		Universit\`a degli Studi di Perugia, Italy\\
		\email{\{emilio.digiacomo,giuseppe.liotta\}@unipg.it}
		\email{alessandra.tappini@studenti.unipg.it} \and
		Roma Tre University, Italy\\
		\email{patrigna@dia.uniroma3.it}\and
		University of Passau, Germany\\
		\email{rutter@fim.uni-passau.de}
	}
	
	\maketitle

\begin{abstract}
We study the \nt planarity testing problem for flat clustered graphs when the maximum size of each cluster is bounded by a constant $k$. We consider both the case when the sides of the matrices to which the edges are incident are fixed and the case when they can be chosen arbitrarily.
We show that \nt planarity testing with fixed sides can be solved in $O(k^{3k+\frac{3}{2}} \cdot n)$ time for every flat clustered graph that can be reduced to a partial 2-tree by collapsing its clusters into single vertices. In the general case, \nt planarity testing with fixed sides can be solved in $O(n)$ time for $k = 2$, but it is NP-complete for any $k > 2$.  \nt planarity testing remains NP-complete also in the free sides model when $k > 4$.
\end{abstract}

\section{Introduction}

Motivated by the need of visually exploring non-planar graphs, hybrid planarity is one of the emerging topics in graph drawing (see, e.g., \cite{addfpr-ilrg-17,bbdlpp-valg-11,ddfp-cnrcg-jgaa-17,hfm-dhvsn-07}). A hybrid planar drawing of a non-planar graph suitably represents in restricted geometric regions those dense subgraphs for which a classical node-link representation paradigm would not be visually effective. These regions are connected by edges that do not cross each other. Different representation paradigms for the dense subgraphs give rise to different types of hybrid planar drawings.

Angelini \emph{et al.}~\cite{addfpr-ilrg-17} consider hybrid planar drawings where dense portions of the graph are represented as intersection graphs of sets of rectangles and study the complexity of testing whether a non-planar graph admits such a representation. In the context of social network analysis, Henry \emph{et al.}~\cite{hfm-dhvsn-07} introduce NodeTrix representations, where the dense subgraphs are represented as adjacency matrices (see Fig.~\ref{fi:nodetrix-demo} for a \nt representation drawn by the online prototype~\cite{giordano-demo}). Batagelj \emph{et al.}~\cite{bbdlpp-valg-11} study the question of minimizing the size of the matrices in a NodeTrix representation of a graph while guaranteeing the planarity of the edges that connect different matrices. While Batagelj \emph{et al.} can choose the subgraphs to be represented as matrices,
Da Lozzo \emph{et al.}~\cite{ddfp-cnrcg-jgaa-17} consider the problem of testing whether a flat clustered graph (i.e.\ a graph with clusters and no sub-clusters) admits a \nt planar representation. In the paper of Da Lozzo \emph{et al.} each cluster must be represented by a different adjacency matrix and the inter-cluster edges are represented as non-intersecting simple Jordan arcs. They prove that NodeTrix planarity testing for flat clustered graphs is NP-hard even in the constrained case where for each matrix it is specified which inter-cluster edges must be incident on the top, on the left, on the bottom, on the right side.

\begin{figure}[tb]
\centering
\includegraphics[width=0.60\textwidth]{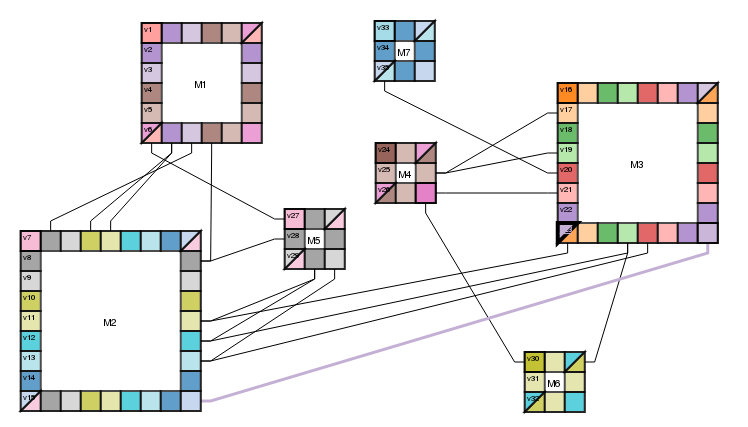}
\caption{A \nt representation with many crossings drawn with the online editor~\cite{giordano-demo} (courtesy of the authors of~\cite{ddfp-cnrcg-jgaa-17}).}\label{fi:nodetrix-demo}
\end{figure}

Motivated by these hardness results, in this paper we study whether NodeTrix planarity testing can be efficiently solved when the size of the clusters is not ``too big''. More precisely, we consider flat clustered graphs whose clusters have size bounded by a fixed parameter $k$ and we want to understand whether the \nt planarity testing problem is fixed parameter tractable, i.e. it can be solved in time $O(f(k)\,n^c)$, where $c$ is a constant and $f(k)$ is a computable function that depends only on $k$. Note that, in some contexts, $k$ is commonly used to denote the number of clusters in a clustered graph; we remark that in the following we denote by $k$ the size of clusters.
Our main results can be listed as follows:
\begin{itemize}

\item We describe an $O(k^{3k+\frac{3}{2}}\cdot n)$-time algorithm to test \nt planarity with fixed sides for every $n$-vertex flat clustered partial $2$-tree $G$. Informally, a flat clustered partial $2$-tree is a flat clustered graph such that by collapsing every cluster into a single vertex we obtain a partial 2-tree. We recall that partial $2$-trees include series-parallel graphs, which are a classical subject of investigation in graph theory and graph drawing, see, e.g.,~\cite{DiGiacomo2006,DBLP:journals/jacm/TakamizawaNS82,DBLP:journals/siamcomp/ValdesTL82,zn-odspgmb-08}.

\item When the flat clustered graph is not a partial 2-tree, \nt planarity testing with fixed sides can still be solved in $O(n)$ time for $k = 2$, but it becomes NP-complete for any larger value of $k$.

\item Finally, we extend the above hardness result to the free sides model and we show that \nt planarity testing remains NP-complete when the maximum cluster dimension is larger than four.

\end{itemize}

From a technical point of view, our linear-time algorithms solve special types of planarity testing problems, where the order of the edges around the vertices is suitably constrained to take into account the fact that each vertex of a matrix $M$ has four copies along the four sides of $M$. It may be worth recalling that Gutwenger {\em et al.}~\cite{gkm-ptoei-08} considered a similar problem. Namely, they studied planarity testing of non-clustered graphs with the additional constraint that the order of the edges around the vertices may not be arbitrarily permuted. The solution by Gutwenger {\em et al.}~\cite{gkm-ptoei-08} is based on modeling the embedding constraints as suitable gadgets of polynomial size that are added to the input graph so to form an enriched graph. The graph has a constrained planar embedding if and only if the enriched graph is planar. For NodeTrix planarity testing and $k=2$, we extend the approach of Gutwenger {\em et al.}~\cite{gkm-ptoei-08} by adding a new gadget, and show that a flat clustered graph is NodeTrix planar if and only if the enriched graph is planar. The gadget that we introduce models an ``embedding synchronization constraint'' between vertices of different triconnected components. Very informally, this constraint expresses the fact that in a NodeTrix planar embedding there are matrices that influence each other in their order of rows and columns and that an order for the rows of  a matrix implies a same order for its columns. By means of this gadget we can express the NodeTrix planarity testing problem for $k=2$ as a 2SAT problem. 

For matrices of size $k > 2$, however, it is unclear how to efficiently solve NodeTrix planarity testing by means of gadgets of polynomial size. This characteristic associates NodeTrix planarity testing for $k>2$ with other known variants of planarity testing, including clustered planarity, where the use of gadgets of polynomial size has been so far an elusive goal. In fact, our linear-time solution for $k >2$ and flat clustered graphs that are partial $2$-trees does not use a 2SAT formulation and it is based on an efficient visit of the block-cut-vertex decomposition tree of the graph.


The rest of the paper is organized as follows. Preliminary definitions are in Section~\ref{se:preliminaries}. Sections~\ref{se:overview} and \ref{se:partial-2-trees} describe a linear-time algorithm for clustered partial 2-trees with bounded cluster size. In Section~\ref{se:general-planar} we show that for general flat clustered graphs and fixed sides \nt planarity testing can be solved in linear time for $k=2$, but it is NP-complete for $k>2$. In Section~\ref{se:free-sides} we extend this completeness result to \nt planarity testing of flat clustered graphs with free sides. Finally, conclusions and open problems can be found in Section~\ref{se:open-problems}.

\section{Preliminaries}\label{se:preliminaries}

We assume familiarity with basic definitions of graph theory and graph drawing (see, e.g.,~\cite{dett-gd-99,Harary69a}). 

A {\em flat clustered graph} $G=(V,E,\mathcal{C})$ is a simple graph with vertex set $V$, edge set $E$, and a partition $\mathcal{C}$ of $V$ into sets $V_1,\dots,V_h$, called {\em clusters};  see Fig.~\ref{fi:example-notlight}. An edge $(u,v) \in E$ with $u \in V_i$ and $v \in V_j$ is an {\em intra-cluster edge} if $i=j$ and it is an {\em inter-cluster edge} if $i\neq j$.

A \emph{\nt representation} of a flat clustered graph $G$ is such that:
\textbf{(i)}~Each cluster $V_i$ with $|V_i|=1$ (called \emph{trivial} cluster) is represented as a distinct point in the plane.
\textbf{(ii)}~Each cluster $V_i$ with $|V_i|>1$ (called \emph{non-trivial} cluster) is represented by a symmetric adjacency matrix
$M_i$ (with $|V_i|$ rows and columns), where $M_i$ is drawn in the plane so that its boundary is a square with sides parallel to the coordinate axes.
\textbf{(iii)}~There is no intersection between two distinct matrices or between a point representing a vertex and a matrix.
\textbf{(iv)}~Each intra-cluster edge of a cluster $V_i$ is represented by the adjacency matrix $M_i$.
\textbf{(v)}~Each inter-cluster edge $(u,v)$ with $u \in V_i$ and $v \in V_j$ is represented by a simple Jordan arc connecting a point on the boundary of matrix $M_i$ with a point on the boundary of matrix $M_j$, where the point on $M_i$ (on $M_j$) belongs to the column or to the row of $M_i$ (resp.\ of $M_j$) associated with~$u$ (resp.\ with~$v$). We remark that we require the adjacency matrices to be symmetric, consistently with previous papers that studied the NodeTrix model~\cite{bbdlpp-valg-11,ddfp-cnrcg-jgaa-17,hfm-dhvsn-07}. 

\begin{figure}[tb]
	\centering
	\subfigure[]{\includegraphics[width=0.36\textwidth,page=2]{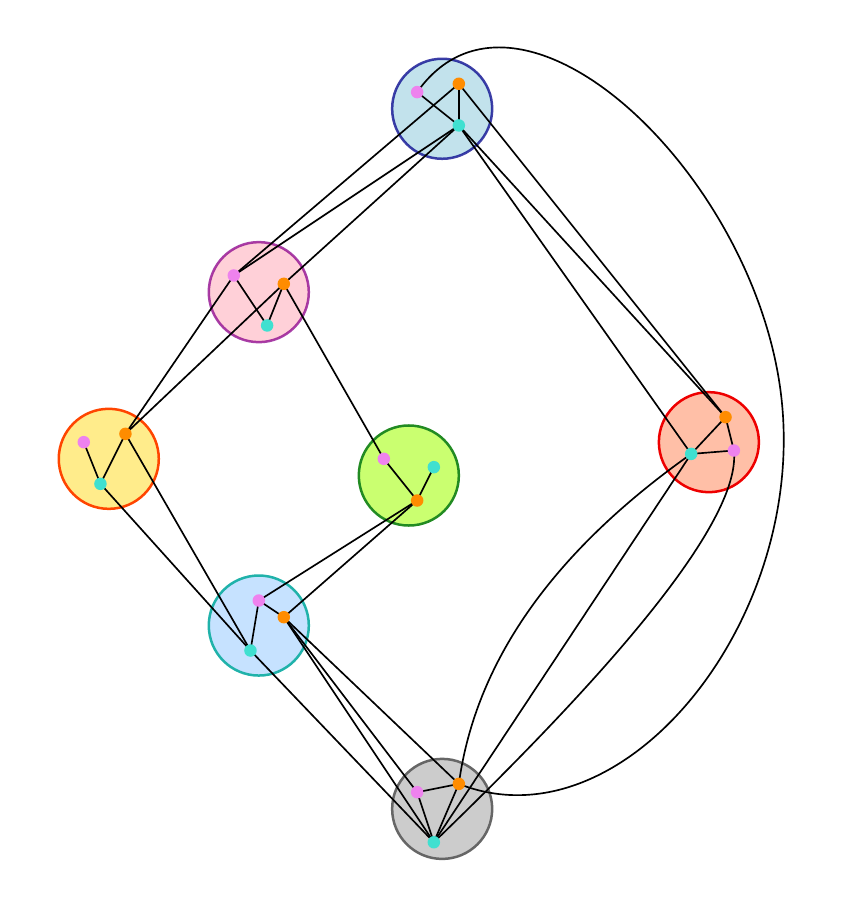}\label{fi:example-notlight}}
	\hfil
	\subfigure[]{\includegraphics[width=0.36\textwidth,page=7]{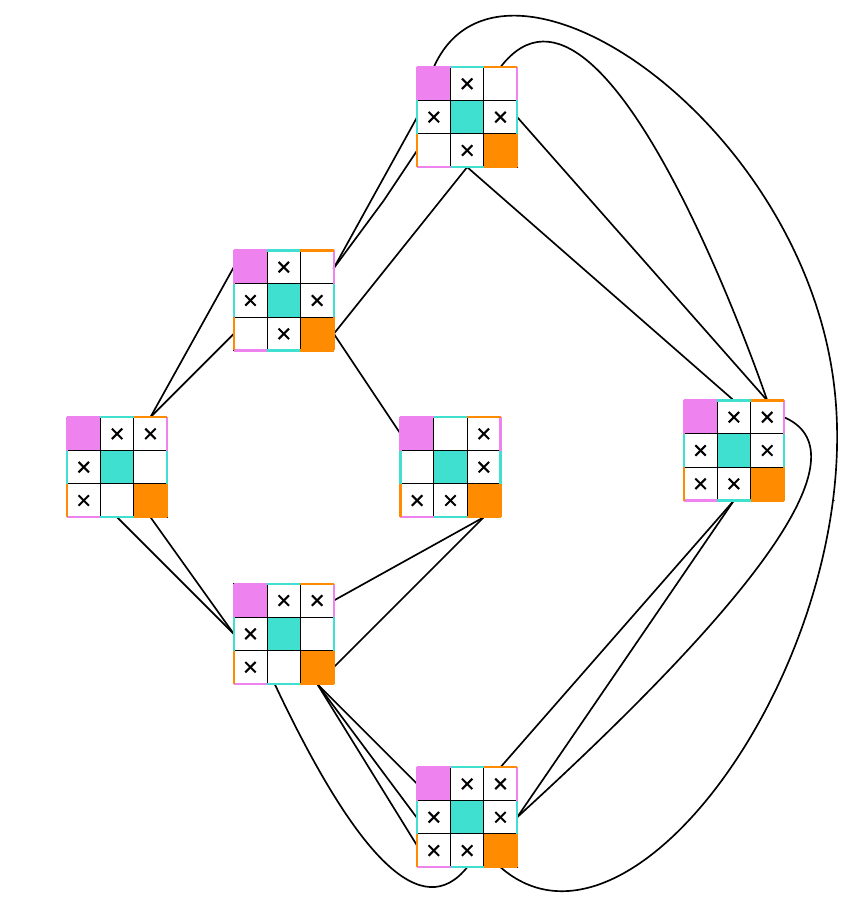}\label{fi:example-nodetrix}}
	\caption{\subref{fi:example-notlight} A flat clustered graph $G$ that is not light. Clusters are represented with different colors. \subref{fi:example-nodetrix} A \nt planar representation of $G$.}\label{fi:clusteredgraph-nodetrix}
\end{figure}

A \nt representation of a flat clustered graph $G$ is \emph{planar} if there is no intersection between any two inter-cluster edges (except possibly at common end-points) nor an intersection between an inter-cluster edge and a matrix.  A flat clustered graph is \emph{\nt planar} if it admits a planar \nt representation. Fig.~\ref{fi:example-nodetrix} shows a \nt planar representation of the flat clustered graph of Fig.~\ref{fi:example-notlight}.

A formal definition of the problem investigated in the paper is as follows.
Let $G=(V,E,\mathcal{C})$ be a flat clustered graph with $n$ vertices and let $k$ be the maximum cardinality of a cluster in $\mathcal{C}$. Clustered graph $G$ is \emph{\nt planar with fixed sides} if it has a \nt planar representation where for each inter-cluster edge, the sides of matrices it attaches to is specified as part of the input; $G$ is \emph{\nt planar with free sides} if the sides of the matrices to which inter-cluster edges attach can be chosen arbitrarily.

Let $M_i$ be the matrix representing cluster $V_i$ in a \nt representation of $G$; let $v$ be a vertex of $V_i$ and let $(u,v)$ be an inter-cluster edge. Edge $(u,v)$ can intersect the boundary of $M_i$ in four points $p_{v,\textrm{\sc{t}}}, p_{v,\textrm{\sc{b}}}, p_{v,\textrm{\sc{l}}}$, and $p_{v,\textrm{\sc{r}}}$ since the row and column that represent $v$ in $M_i$ intersect the four sides of the boundary of $M_i$. We call these points the \emph{top copy}, \emph{bottom copy}, \emph{left copy}, and \emph{right copy} of $v$ in $M_i$, respectively.

A {\em side assignment for $V_i \in \mathcal{C}$} specifies for each inter-cluster edge whether the edge must attach to the matrix $M_i$ representing $V_i$ in its top, left, right, or bottom side. More precisely, a side assignment is a mapping $\phi_i$: $\bigcup_{j\neq i}{E_{i,j}} \rightarrow \{\textrm{\sc{t}}, \textrm{\sc{b}}, \textrm{\sc{l}}, \textrm{\sc{r}}\}$, where $E_{i,j}$ is the set of inter-cluster edges between the clusters $V_i$ and $V_j$ ($V_i$ and $V_j$ are \emph{adjacent} if $E_{i,j}\neq \emptyset$).
A {\em side assignment for $\mathcal{C}$} is a set $\Phi$ of side assignments for each $V_i \in \mathcal{C}$.

We denote as $G=(V,E,\mathcal{C},\Phi)$ a flat clustered graph $G=(V,E,\mathcal{C})$ with a given side assignment~$\Phi=\{\phi_1,\phi_2,\dots, \phi_{|\mathcal{C}|}\}$.
Let $\Gamma$ be a \nt representation of $G$ such that, for every inter-cluster edge $e=(u,v) \in E$ with $u \in V_i$ and $v \in V_j$, the incidence points of $e$ with the matrices $M_i$ and $M_j$ representing $V_i$ and $V_j$ in $\Gamma$ are exactly the points $p_{u,{\phi_i(e)}}$ and $p_{v,{\phi_j(e)}}$, respectively. We call $\Gamma$ a \nt representation of $G$ \emph{consistent with} $\Phi$.
We say that $G=(V,E,\mathcal{C},\Phi)$ is {\em \nt planar} if it admits a \nt planar representation consistent with~$\Phi$.

A flat clustered graph is \emph{light} if no inter-cluster edge has both its end-vertices belonging to non-trivial clusters and no trivial cluster has more than one inter-cluster edge incident to the same non-trivial cluster.
Note that a light flat clustered graph can contain an inter-cluster edge that has both its end-vertices belonging to trivial clusters.
A \emph{$1$-subdivision} of an inter-cluster edge $e=(u,v)$ of a flat clustered graph $G=(V,E,\mathcal{C})$ replaces $e$ by a path $u_0=u, u_1, u_2=v$ and defines a new flat clustered graph $G'=(V',E',\mathcal{C}')$, where $V' = V \cup \{u_1\}$, $E' = E \setminus e \cup \{(u_0,u_1),(u_1,u_2)\}$, and $\mathcal{C}'= \mathcal{C} \cup \{u_1\}$. The \emph{light reduction} of $G$ is the flat clustered graph $G'$ obtained by performing a $1$-subdivision of every inter-cluster edge of $G$. Fig.~\ref{fi:example-light} illustrates the light reduction $G'$ of the flat clustered graph $G$ in Fig.~\ref{fi:example-notlight}.

\begin{figure}[tb]
	\centering
	\subfigure[]{\includegraphics[width=0.36\textwidth,page=3]{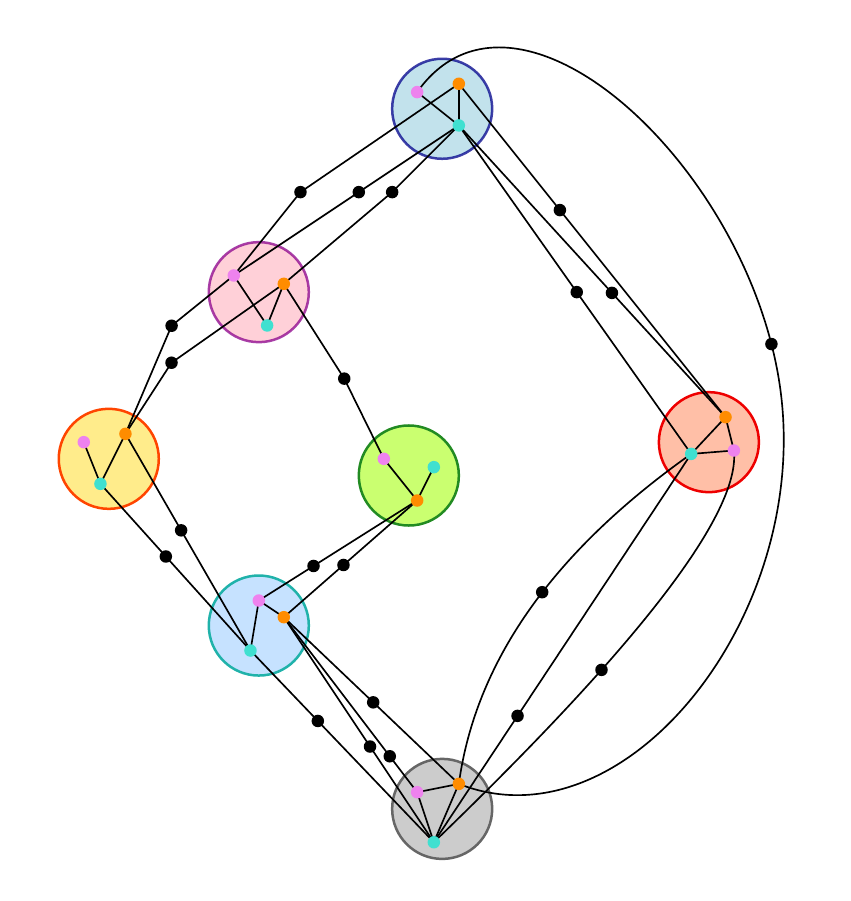}\label{fi:example-light}}
	\hfil
	\subfigure[]{\includegraphics[width=0.2\textwidth,page=5]{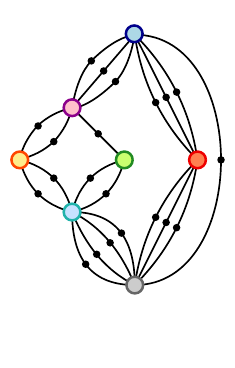}\label{fi:example-frame}}
	\caption{\subref{fi:example-light} The light reduction $G'$ of the graph $G$ in Fig.~\ref{fi:example-notlight}. \subref{fi:example-frame} The frame $F$ of $G'$.}\label{fi:light}
\end{figure}

A consequence of Theorem 1 in~\cite{dlpt-pkng-17}
about the edge density of \nt planar graphs,
is that the light reduction $G'$ of a \nt planar flat clustered graph $G$ has $O(|V|)$ vertices and $O(|V|)$ inter-cluster edges.

\begin{property}\label{pr:lightness}
A flat clustered graph $G$ is \nt planar if and only if its light reduction $G'$ is \nt planar.
\end{property}

Based on Property~\ref{pr:lightness}, in the remainder we shall assume that flat clustered graphs are always light and we call them clustered graphs, for short.

The \emph{frame} of a clustered graph $G=(V,E,\mathcal{C})$ is the graph $F$ obtained by collapsing each cluster $V_i \in \mathcal{C}$, with $|V_i| > 1$, into a single vertex $c_i$ of $F$, called the \emph{representative vertex of $V_i$ in $F$}. Let $c_i$ and $c_j$ be the two representative vertices of $V_i$ and $V_j$ in $F$, respectively. For every inter-cluster edge connecting a vertex of $V_i$ to a vertex of $V_j$ in $G$ there is an edge in $F$ connecting $c_i$ and $c_j$. Observe that the frame graph $F$ of $G$ is in general a multigraph; however, $F$ is simple when $G$ is light. Fig.~\ref{fi:example-frame} shows the frame of the graph in Fig.~\ref{fi:example-light}.

Since the \nt planarity of a clustered graph implies the planarity of its frame graph, we will test \nt planarity only on those clustered graphs that have a planar frame.

A \emph{2-tree} is a graph recursively defined as follows: (i) an edge is a 2-tree; (ii) the graph obtained by adding a vertex $v$ to a 2-tree $G$ and by connecting $v$ to two adjacent vertices of $G$ is a 2-tree.
A (planar) graph is a \emph{partial 2-tree} if it is a subgraph of a (planar) 2-tree.
A biconnected partial 2-tree is a \emph{series-parallel graph}. A clustered graph is a \emph{partial 2-tree} if its frame is a partial 2-tree. We will sometimes talk about series-parallel clustered graphs when their frames are series parallel.

In our testing algorithm we will make use of block-cut-vertex trees and of SPQR-trees.  
We thus conclude this section by recalling their definitions. The \emph{block-cut-vertex tree} of a connected graph $G$ is a tree whose nodes are the \emph{blocks} (i.e. the biconnected components) and the cut-vertices of $G$. There exists an edge between a block $B$ and a cut-vertex $v$ if $v$ belongs to $B$.

\medskip
\noindent \textbf{SPQR-tree.}
Let $G$ be a simply biconnected graph. A \emph{separation pair} is a pair of vertices whose removal disconnects $G$. A \emph{split pair} is either a separation pair or a pair of adjacent vertices. 
A \emph{split component} of a split pair $\{u,v\}$ is either an edge $(u,v)$ or a maximal subgraph $G_{uv} \subset G$ such that $\{u,v\}$ is not a split pair of $G_{uv}$. 
Vertices $\{u,v\}$ are the \emph{poles} of $G_{uv}$.  A split pair $\{s',t'\}$ of $G$ is maximal with respect to a different split pair $\{s,t\}$ of $G$, if for every other split pair $\{s^*,t^*\}$ of $G$, there is a split component that includes the vertices $s',t',s,t$. 
The \emph{SPQR-tree} $T$ of $G$ with respect to an edge $e$ is a rooted tree that describes a recursive decomposition of $G$ induced by its split pairs~\cite{dt-olpt-96}. 
In what follows, we call \emph{nodes} the vertices of $T$, to distinguish them from the vertices of $G$. The nodes of $T$ are of four types S, P, Q, or R. 
Each node $\mu$ of $T$ has an associated biconnected multigraph called the \emph{skeleton of $\mu$} and denoted as $\skel(\mu)$, which contains a distinguished edge, called the \emph{reference edge}, between the two poles of the corresponding split component.  
At each step, given the current split component $G_\mu$, its split pair $\{s,t\}$, and a node $\nu$ in $T$, the node $\mu$ of the tree corresponding to $G_\mu$ is introduced and attached to its parent vertex $\nu$, while the decomposition possibly recurs on some split component of $G_\mu$.  Graph $G_\mu$ is called the \emph{pertinent graph} of $\mu$.
At the beginning of the decomposition the parent of $\mu$ is a Q-node corresponding to $e=(u,v)$, $G_\mu = G \setminus e$, and $\{s,t\} = \{u,v\}$.

\textbf{Base case}: $G_\mu$ consists of a single edge between $s$ and $t$. Then, $\mu$ is a Q-node whose skeleton is $G_\mu$ itself plus the reference edge between $s$ and $t$.  
 
\textbf{Parallel case}: The split pair $\{s,t\}$ has $G_1,\dots,G_k$ ($k \geq 2$) as split components.
Then, $\mu$ is a P-node whose skeleton is a set of $k+1$ parallel edges between $s$ and $t$, one for each split component $G_i$ plus the reference edge between $s$ and $t$.  
The decomposition recurs on $G_1,\dots,G_k$ with $\mu$ as parent node. 

\textbf{Series case}: $G_\mu$ is not biconnected and it has at least one cut-vertex (a vertex whose removal disconnects $G_\mu$).  
Then, $\mu$ is an S-node whose skeleton is defined as follows. 
Let $v_1,\dots,v_{k-1}$, where $k \geq 2$, be the cut vertices of $G_\mu$. 
The skeleton of $\mu$ is a path consisting of the edges $e_1,\dots,e_k$, where $e_i= (v_{i-1},v_i)$, $v_0=s$ and $v_k=t$, plus the reference edge between $s$ and $t$ which makes the path a cycle. 
The decomposition recurs on the split components corresponding to $e_1,\dots,e_k$ with $\mu$ as parent node. 

\textbf{Rigid case}: None of the other cases is applicable. 
Let $\{s_1,t_1\},\dots,\{s_k,t_k\}$ be the maximal split pairs of $G$ with respect to $\{s,t\}$ ($k \geq 1$) such that $\{s_i,t_i\}$ belongs to $G_\mu$, for $i=1,\dots,k$. 
Then $\mu$ is an R-node whose skeleton is the graph obtained from $G_\mu$ as follows. Connect each pair $\{s_i,t_i\}$ with an edge if not already adjacent, connect the poles $s$ and $t$, and finally remove all vertices other than the poles and the pairs $\{s_i,t_i\}$.
The decomposition recurs on each $G_i$ with $\mu$ as parent node. 

\section{\nt Representations and Wheel Reductions}\label{se:overview}

The linear-time algorithms described in Sections~\ref{se:partial-2-trees} and~\ref{se:general-planar} are based on decomposing the planar frame $F$ of a clustered graph $G=(V,E,\mathcal{C},\Phi)$ into its biconnected components and storing them into a block-cut-vertex tree. We process each block of $F$ by using an SPQR decomposition tree that is rooted at a reference edge and visited from the leaves to the root. For each visited node $\mu$ of the decomposition tree of a block of $F$, we test whether the subgraph of $G$ whose frame is the pertinent graph of $\mu$ satisfies the planar constraints imposed by the side assignment on the inter-cluster edges. A key ingredient to efficiently perform the test at $\mu$ is the notion of \emph{wheel replacement}.

Let $G=(V,E,\mathcal{C},\Phi)$ be a clustered graph with side assignment $\Phi$ and let $V_i \in \mathcal{C}$ be a cluster with $k>1$ vertices. $V_i$ admits $k!$ permutations of its vertices and we associate a suitable graph to each such permutation. Let $\pi_i = v_0, v_1, \dots, v_{k-1}$ be a permutation of the vertices of $V_i$.
The \emph{wheel of $V_i$ consistent with $\pi_i$} is the wheel graph consisting of a vertex $v$ of degree $4k$ adjacent to the vertices of an oriented cycle $v_{0,\textrm{\sc{t}}},$ $v_{1,\textrm{\sc{t}}},$ $\dots, v_{k-1,\textrm{\sc{t}}}$, $v_{0,\textrm{\sc{r}}}$, $v_{1,\textrm{\sc{r}}}, \dots, v_{k-1,\textrm{\sc{r}}}$, $v_{k-1,\textrm{\sc{b}}}$, $v_{k-2,\textrm{\sc{b}}}, \dots, v_{0,\textrm{\sc{b}}}$, $v_{k-1,\textrm{\sc{l}}}$, $v_{k-2,\textrm{\sc{l}}}$, $\dots, v_{0,\textrm{\sc{l}}}$ where each edge of the cycle is oriented forward. Intuitively, this oriented cycle will be embedded clockwise to encode the constraints induced by a matrix $M_i$ representing $V_i$ when its left-to-right order of columns is $\pi_i$.
More precisely, a \emph{wheel replacement of cluster $V_i$ consistent with $\pi_i$} is the clustered graph obtained as follows: (i) remove $V_i$ and all the inter-cluster edges incident to $V_i$; (ii) insert the wheel $W_i$ of $V_i$ consistent with $\pi_i$; and (iii) for each inter-cluster edge $e=(u,v_j)$, with $v_j \in V_i$, insert edge $(u,v_{j,\phi_i(e)})$ incident to $W_i$.
We call edge $(u,v_{j,\phi_i(e)})$ the \emph{image} of edge~$e=(u,v_j)$.

Let $G=(V,E,\mathcal{C},\Phi, \Pi)$ be a clustered graph with side assignment $\Phi$ where
$\Pi$ is a set of permutations $\{\pi_1, \pi_2, \dots, \pi_{|\mathcal{C}|}\}$, one for each cluster $V_i$ (with $i=1, \dots, |\mathcal{C}|$). We call $\Pi$ the \emph{permutation assignment} of $G$ and we say that $G$ is \emph{\nt planar with side assignment $\Phi$ and permutation assignment $\Pi$} if $G$ admits a \nt planar representation with side assignment $\Phi$ where for each matrix $M_i$ the permutation of its columns is $\pi_i$.
The \emph{wheel reduction of $G$ consistent with $\Pi$} is the graph obtained by performing a wheel replacement of $V_i \in \mathcal{C}$ consistent with $\pi_i$ for each $i=1, \dots, |\mathcal{C}|$.

\begin{theorem}\label{th:wheel-reduction}
Let $G=(V,E,\mathcal{C},\Phi, \Pi)$ be a clustered graph with side assignment $\Phi$ and permutation assignment $\Pi$. $G$ is \nt planar if and only if the planar wheel reduction of $G$ admits a planar embedding where the external oriented cycle of each wheel $W_i$ is embedded clockwise.
\end{theorem}
\begin{proof}
	If $G$ is \nt planar, we construct a planar embedding of a wheel reduction of $G$ where the external oriented cycle of each wheel is embedded clockwise as follows.
	Let $\Gamma$ be a \nt planar representation of $G$. We replace each matrix $M_i$ representing a cluster $V_i \in \mathcal{C}$ by the wheel $W_i$ of $V_i$ consistent with the permutation $\pi_i \in \Pi$ of $V_i$. Also $W_i$ is embedded in such a way that a forward traversal of its external cycle is a clockwise traversal of the cycle. Every inter-cluster edge $e=(u,v_j)$, with $v_j \in V_i$, is incident to the vertex $v_{j,\phi_i(e)}$ of the wheel $W_i$. Also, for all $j= 0, 1, \dots, k-1$ and for all $\textrm{\sc{x}} \in \{ \textrm{\sc{t}}, \textrm{\sc{b}}, \textrm{\sc{l}}, \textrm{\sc{r}} \}$, the cyclic order of the inter-cluster edges incident to $v_{j,\textrm{\sc{x}}}$ in $W_i$ is the same as the cyclic order of the inter-cluster edges incident to $p_{v_j,\textrm{\sc{x}}}$ in $M_i$. It is immediate to see that, since no two inter-cluster edges cross in $\Gamma$, no two edges cross in the constructed embedding of the wheel reduction of $G$.
	
	Conversely, suppose that we are given a planar embedding of the wheel reduction of $G$ where the external oriented cycle of each wheel is embedded clockwise. We show how to construct a \nt planar representation of $G$. For each wheel $W_i$ we remove the center vertex of the wheel and insert a matrix $M_i$ inside the created face. We now morph every vertex $v_{j,\phi_i(e)}$ of the external cycle of $W_i$ to point $p_{v_j,\phi_i(e)}$ in $M_i$ and maintain around $p_{v_j,\phi_i(e)}$ the cyclic order of the inter-cluster edges incident to $v_{j,\phi_i(e)}$ in the planar embedding of the wheel reduction.
	\qed
\end{proof}

Fig.~\ref{fi:example-matrices} and Fig.~\ref{fi:example-wheels} show respectively a \nt planar representation of the flat clustered graph in Fig.~\ref{fi:example-light} and the corresponding wheel reduction with its planar embedding.

\begin{figure}[tb]
	\centering
	\subfigure[]{\includegraphics[width=0.45\textwidth,page=8]{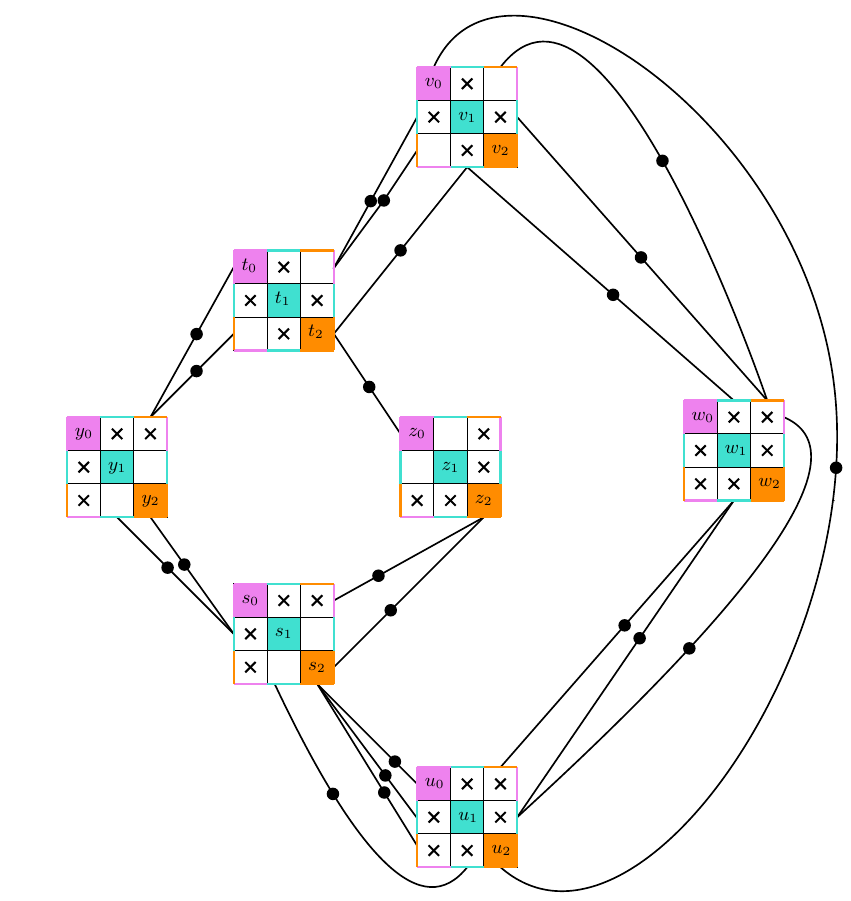}\label{fi:example-matrices}}
	\hfil
	\subfigure[]{\includegraphics[width=0.45\textwidth,page=9]{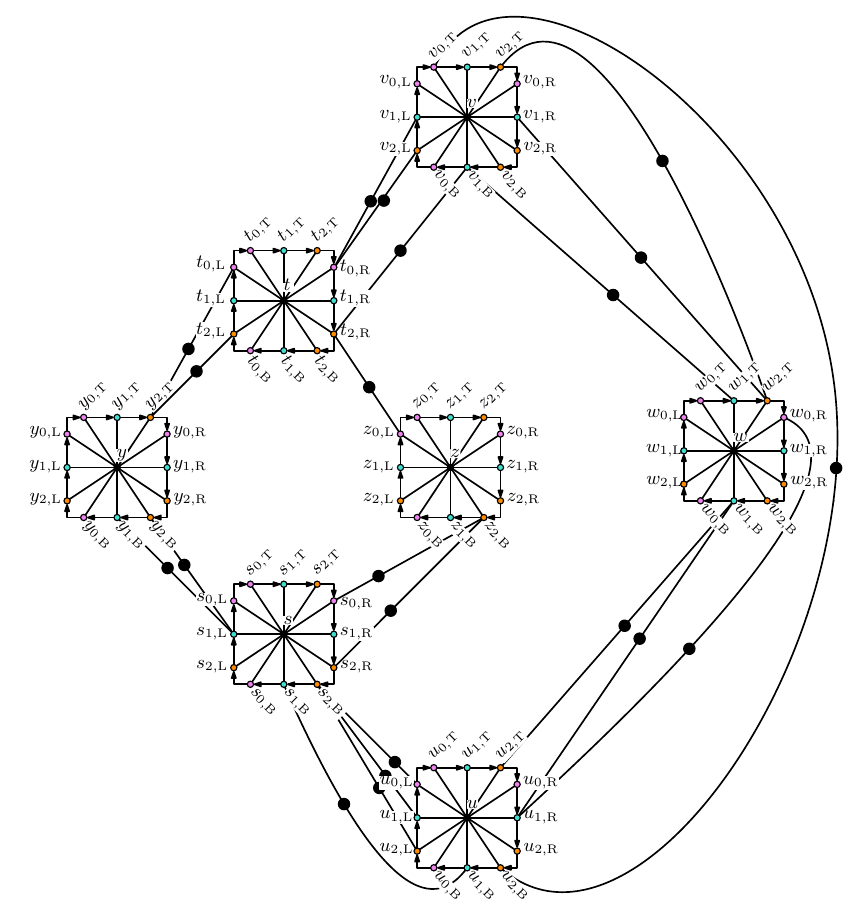}\label{fi:example-wheels}}
	\caption{\subref{fi:example-matrices} A \nt planar representation of the  graph in Fig.~\ref{fi:example-light}. \subref{fi:example-wheels} The planar embedding of the corresponding wheel reduction.}
\end{figure}

Based on Theorem~\ref{th:wheel-reduction}, we can test the graph $G=(V,E,\mathcal{C},\Phi)$ for \nt planarity by exploring the space of the possible permutation sets $\Pi$ and the corresponding wheel reductions in search of a \nt planar $G=(V,E,\mathcal{C},\Phi,\Pi)$.
Note that, if the maximum size of a cluster is given as a parameter $k$, every cluster $V_i$ can be replaced by $k!$ wheel graphs, one for each possible permutation of the vertices of $V_i$. In order to test planarity, for any such wheel replacement $W_i$, the cyclic order of the inter-cluster edges incident to the same vertex of $W_i$ can be arbitrarily permuted. While each wheel reduction yields an instance of constrained planarity testing that can be solved with the linear-time algorithm described in~\cite{gkm-ptoei-08}, a brute-force approach that repeats this algorithm on each possible wheel reduction may lead to testing planarity on $(k!)^{|\mathcal{C}|}$ different instances.
Instead, for each visited node $\mu$ of the decomposition tree $T$ we compute a succinct description of the possible \nt planar representations of the subgraph $G_\mu$ of $G$ represented by the subtree of $T$ rooted at $\mu$. This is done by storing for the poles of $\mu$ those pairs of wheel graphs that are compatible with a \nt planar representation of $G_\mu$.
How to efficiently compute such a succinct description will be the subject of the next sections.

\section{Testing \nt Planarity for Partial 2-Trees}\label{se:partial-2-trees}

In this section we prove that \nt planarity testing with fixed sides can be solved in linear time for a clustered graph $G=(V,E,\mathcal{C},\Phi)$ when the maximum size of any cluster of $\mathcal{C}$ is bounded by a constant and the frame graph is a partial 2-tree.
This contrasts with the NP-hardness of \nt planarity testing with fixed sides proved in~\cite{ddfp-cnrcg-jgaa-17} in the case where the size of the clusters is unbounded.

We first study the case of a clustered graph whose frame graph is a series-parallel graph, i.e., it is biconnected and its SPQR decomposition tree only has Q-, P-, and S-nodes. We refer to such trees as SPQ decomposition trees.
We then consider the case of partial 2-trees, i.e., graphs whose biconnected components are series-parallel.

%
%
\subsection{Series-Parallel Frame Graphs}\label{sse:series-parallel}

In this section we prove that \nt planarity testing with fixed sides can be solved in $O(k^{3k+\frac{3}{2}}\cdot n)$ time for clustered graphs whose frame graphs are series-parallel and have cluster size at most $k$.

Let $G=(V,E,\mathcal{C},\Phi)$ be a series-parallel clustered graph with side assignment $\Phi$ and let $F$ be its frame graph. Let $T$ be the SPQ decomposition tree of $F$ rooted at any Q-node; see Fig.~\ref{fi:sp-frame} and Fig.~\ref{fi:sp-decomposition-tree} for an example. To simplify the description and without loss of generality, we assume that every S-node of $T$ has exactly two children. Let $\mu$ be a node of $T$, and let $s_\mu$ and $t_\mu$ be the poles of $\mu$ (refer to Fig.~\ref{fi:frame-spq}).
Consider the pertinent graph $F_\mu$ represented by the subtree of $T$ rooted at $\mu$ and let $v_\mu$ be a pole of $\mu$ ($v_\mu \in \{s_\mu, t_\mu\}$).
Pole $v_\mu$ in the frame graph $F$ may correspond to a non-trivial cluster $V_i$ of $\mathcal{C}$. In this case, we call $v_\mu$ a \emph{non-trivial pole of $\mu$} and cluster $V_i$ the \emph{pertinent cluster} of $v_\mu$.
\begin{figure}[tb]
	\centering
	\subfigure[]{\includegraphics[width=0.22\textwidth,page=4]{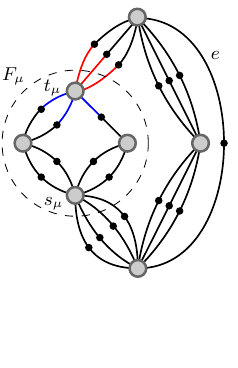}\label{fi:sp-frame}}
	\hfil
	\subfigure[]{\includegraphics[width=0.77\textwidth,page=6]{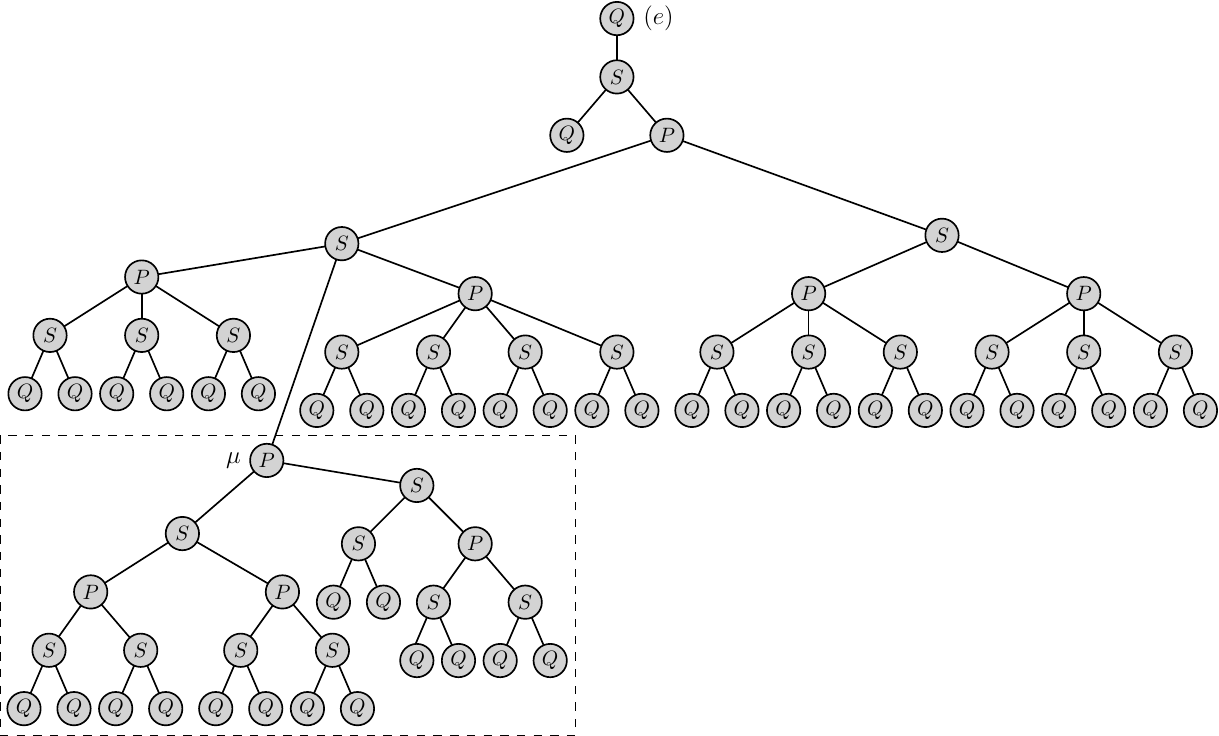}\label{fi:sp-decomposition-tree}}
	\caption{\subref{fi:sp-frame} The frame graph $F$ of the graph $G$ in Fig.~\ref{fi:example-light}. The pertinent graph $F_\mu$ of a node $\mu$ and its poles $s_\mu$ and $t_\mu$ are highlighted. Intra- and extra-component edges are colored blue and red, respectively. \subref{fi:sp-decomposition-tree} The SPQ decomposition tree $T$ of $F$ rooted at edge $e$. The subtree of $T$ rooted at $\mu$ is highlighted.}\label{fi:frame-spq}
\end{figure}
The edges of $F_\mu$ incident to $v_\mu$ are the \emph{intra-component edges of $v_\mu$}. The other edges of $F$ incident to $v_\mu$ are the \emph{extra-component edges of $v_\mu$}.
The intra-component edges of $t_\mu$ are colored blue in Fig.~\ref{fi:sp-frame}, while the extra-component edges of $t_\mu$ are colored red.
Each intra-component (extra-component) edge of $v_\mu$ corresponds to an inter-cluster edge $e'$ of $G$ incident to one vertex of the pertinent cluster $V_\mu$ of $v_\mu$. We call $e'$ an \emph{intra-component edge (extra-component edge) of $V_\mu$}. We associate $k!$ wheel graphs to each non-trivial pole $v_\mu$ of $\mu$. Each of them is a wheel replacement of the pertinent cluster of $v_\mu$, consistent with one of the $k!$ permutations of its vertices.

Let $v_\mu$ be a non-trivial pole of $\mu$, let $V_\mu$ be the pertinent cluster of $v_\mu$, let $\pi_\mu$ be a permutation of the vertices of $V_\mu$, and let $W_\mu$ be the wheel replacement of $V_\mu$ consistent with $\pi_\mu$.
Every edge $e$ incident to $W_\mu$ such that $e$ is the image of an inter-cluster edge $e'$ of $G$ is labeled either \texttt{int} or \texttt{ext}, depending on whether $e'$ is an intra-component or an extra-component edge of $V_\mu$.
A vertex $w$ of the external cycle of $W_\mu$ is assigned one label of the set $\{\texttt{void}, \texttt{int},\texttt{ext}, \texttt{int-ext}\}$ as follows. Vertex $w$ is labeled \texttt{void} if no edge incident to $w$ is the image of an inter-cluster edge. Vertex $w$ is labeled \texttt{int} (resp.\ \texttt{ext}) if we have a label \texttt{int} (resp.\ \texttt{ext}) on every edge $e$ incident to $w$ such that $e$ is the image of an inter-cluster edge. Otherwise, vertex $w$ is labeled \texttt{int-ext}. See Fig.~\ref{fi:wheels} for an example concerning the wheel $W_{t_{\mu}}$ of Fig.~\ref{fi:p-wheels}; the dashed curve of Fig.~\ref{fi:p-wheels} shows the subgraph of the wheel reduction corresponding to the pertinent graph $F_{\mu}$ of Fig.~\ref{fi:sp-frame}.

\begin{figure}[tb]
	\centering
	\subfigure[]{\includegraphics[width=0.45\textwidth,page=1]{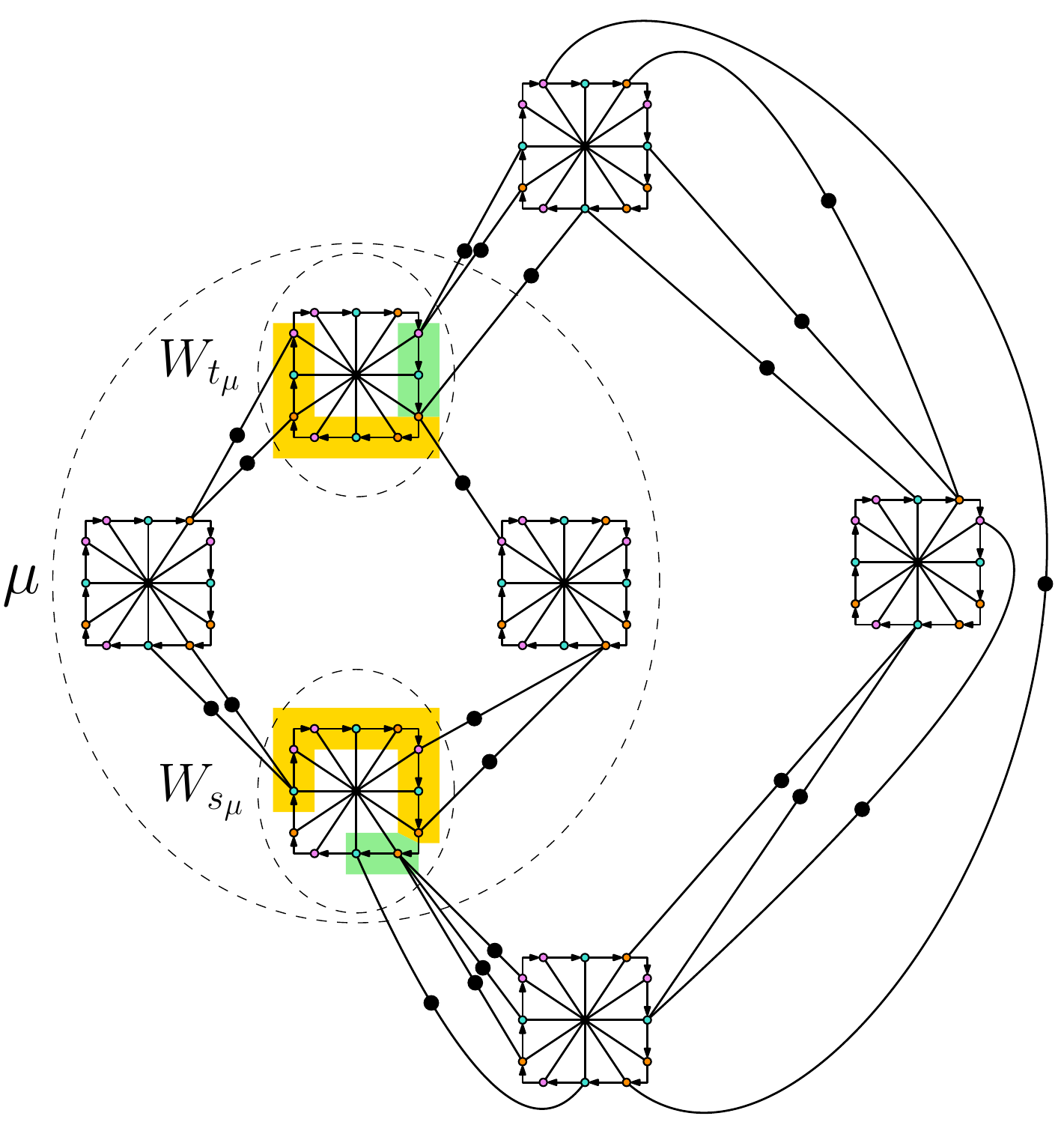}\label{fi:p-wheels}}
	\hfil
	\subfigure[]{\includegraphics[width=0.23\textwidth,page=2]{sp-wheels.pdf}\label{fi:wheels}}
	\caption{(a) The wheel reduction of graph $G'$ in Fig.~\ref{fi:example-light}; the complete internal and external sequences for a pair of poles are also highlighted. (b) Labeling of the vertices of $W_{t_{\mu}}$.}\label{fi:sp}
\end{figure}

A clockwise sequence $v_0, v_1, \dots, v_j$ of vertices of the external cycle of $W_\mu$ is an \emph{external sequence of pole $v_\mu$ consistent with $\pi_\mu$} if $v_0$ and $v_j$ are labeled either \texttt{ext} or \texttt{int-ext} and all the other vertices of the sequence are labeled either \texttt{void} or \texttt{ext}. An external clockwise sequence of pole $v_\mu$ is \emph{complete} if it contains all the vertices of $W_\mu$ that are labeled \texttt{ext} and \texttt{int-ext}. Note that a complete external sequence may contain many \texttt{void} vertices but no \texttt{int} vertex.
\emph{Internal} and \emph{complete internal} sequences of pole $v_\mu$ are defined analogously; see Fig.~\ref{fi:sp} for an illustration.
Observe that a complete internal sequence and a complete external sequence of $v_\mu$ may not exist when vertices labeled \texttt{int} and vertices labeled \texttt{ext} alternate more than twice when traversing clockwise the external cycle of $W_\mu$, or when three vertices are labeled \texttt{int-ext}.
A special case is when $W_\mu$ has exactly two vertices $w_1$ and $w_2$ labeled \texttt{int-ext} and all other vertices are \texttt{void}. In this case, the clockwise sequence from $w_1$ to $w_2$ and the clockwise sequence from $w_2$ to $w_1$ are both complete internal and complete external sequences.

In order to test $G=(V,E,\mathcal{C},\Phi)$ for \nt planarity, we implicitly take into account all possible permutation assignments $\Pi$ by considering, for each non-trivial pole $w_\mu$ of each node $\mu$ of $T$, its $k!$ possible wheels and by computing their complete internal and complete external sequences.
We visit the SPQ decomposition tree $T$ from the leaves to the root and equip each node $\mu$ of $T$ with information regarding the complete internal and complete external sequences of its non-trivial poles.
Let $\mu$ be an internal node of $T$, let $v_\mu$ be a non-trivial pole of $\mu$, let $\pi_{v_\mu}$ be a permutation of the pertinent cluster $V_\mu$ of $v_\mu$, and let $W_\mu$ be the wheel of $V_\mu$ consistent with $\pi_{v_\mu}$.
We denote as $ISeq(\mu, v_{\mu},\pi_{v_\mu})$ the complete internal sequence of $v_\mu$ consistent with $\pi_{v_\mu}$ in pole $\mu$ and as $ESeq(\mu, v_{\mu},\pi_{v_\mu})$ the complete external sequence of $v_{\mu}$ consistent with $\pi_{v_\mu}$ in pole $\mu$.
We distinguish between the different types of nodes of $T$.

\smallskip\noindent
\textbf{Node $\mu$ is a Q-node.}
Since $G$ is light, at most one of its poles is non-trivial. Let $e$ be an edge of $F$ that is the pertinent graph of $\mu$. One end-vertex of $e$ is the representative vertex in $F$ of the pertinent cluster of the non-trivial pole $v_\mu$. In fact, edge $e$ corresponds to an edge $e'=(u,z)$ of $G$ such that $u \in V_\mu$ and $z$ is a trivial cluster. The side assignment $\phi_{v_\mu}$ defines whether $e$ is incident to the top, bottom, left, or right copy $u_W$ of $u$ in the wheel $W_\mu$ of $V_\mu$. For any possible permutation $\pi_{v_\mu}$ we have $ISeq(\mu, v_\mu,\pi_{v_\mu})=u_W$. If $u_W$ is labeled \texttt{int-ext}, then $ESeq(\mu, v_\mu,\pi_{v_\mu})$ is the external cycle of $W_\mu$ starting at $u_W$ and ending at $u_W$. Otherwise, traverse the external cycle of $W_\mu$ starting at $u_W$ and following the direction of the edges; $ESeq(\mu, v_\mu,\pi_{v_\mu})$ consists of all the encountered vertices from the first labeled \texttt{ext} to the last labeled \texttt{ext}.

\smallskip\noindent
\textbf{Node $\mu$ is a P-node.}
Let $\nu_0, \nu_1, \dots, \nu_{h-1}$ be the children of $\mu$. Observe that $v_\mu$ is a non-trivial pole also for the children $\nu_0, \nu_1, \dots, \nu_{h-1}$ of $\mu$. We consider every permutation $\pi_{v_\mu}$ such that $\nu_0, \nu_1, \dots, \nu_{h-1}$ have both a complete internal sequence and a complete external sequence compatible with $\pi_{v_\mu}$.
The complete internal sequence of $v_\mu$ consistent with $\pi_{v_\mu}$ is the union of the complete internal sequences of the children $\nu_0, \nu_1, \dots, \nu_{h-1}$, that is $ISeq(\mu, v_\mu,\pi_{v_\mu})= \displaystyle\cup_{i=0}^{h-1} ISeq(\nu_i,v_\mu,\pi_{v_\mu})$.

To determine the complete external sequence of $v_\mu$ consistent with $\pi_{v_\mu}$ we consider the intersection of the complete external sequences of the children of $\mu$.
If this intersection consists of exactly one sequence of consecutive vertices, then $ESeq(\mu, v_\mu,\pi_{v_\mu})= \displaystyle\cap_{i=0}^{h-1} ESeq(\nu_i,v_\mu,\pi_{v_\mu})$. Otherwise (i.e., the intersection is empty or it consists of more than one sequence of consecutive vertices), $v_\mu$ does not have a complete external sequence consistent with $\pi_{v_\mu}$.

\smallskip\noindent
\textbf{Node $\mu$ is an S-node.}
Let $\nu$ be the child of $\mu$ that shares the pole $v_\mu$ with $\mu$.
We consider every permutation $\pi_{v_\mu}$ such that $\nu$ has both $ISeq(\nu, v_\mu,\pi_{v_\mu})$ and $ESeq(\nu, v_\mu,\pi_{v_\mu})$.
The complete internal (external) sequence of $v_\mu$ consistent with $\pi_{v_\mu}$ is $ISeq(\mu, v_\mu,\pi_{v_\mu})=ISeq(\nu, v_\mu,\pi_{v_\mu})$ ($ESeq(\mu, v_\mu,\pi_{v_\mu})=ESeq(\nu, v_\mu,\pi_{v_\mu})$).

To test $G$ for \nt planarity we execute a bottom-up traversal of $T$ and, for each node $\mu$ with poles $s_\mu$ and $t_\mu$, we check whether each possible pair $(\pi_{s_\mu}, \pi_{t_\mu})$ induces complete internal and external sequences for $s_\mu$ and $t_\mu$ that are `compatible' with a planar embedding of the wheel reduction of $G$. If this is the case, by Theorem~\ref{th:wheel-reduction}, $G$ is \nt planar, otherwise we reject $G$.

More formally, let $\pi_{s_\mu}$ ($\pi_{t_\mu}$, respectively) be a permutation such that $s_\mu$ ($t_\mu$, respectively) has both a complete internal sequence and a complete external sequence compatible with $\pi_{s_\mu}$ ($\pi_{t_\mu}$, respectively). We say that $(\pi_{s_\mu}, \pi_{t_\mu})$ is a \emph{compatible pair of permutations for $\mu$} if either one of the poles is a trivial pole or one of the following cases applies.

\smallskip\noindent
\textbf{Node $\mu$ is a Q-node.}
In this case all $k!$ possible pairs of permutations for $s_\mu$ or $t_\mu$ (recall that only one of them is non-trivial) are compatible for $\mu$.

\smallskip\noindent
\textbf{Node $\mu$ is a P-node.}
Let $\nu_0, \nu_1, \dots, \nu_{h-1}$ be the children of $\mu$. Consider a pair of permutations $(\pi_{s_\mu}, \pi_{t_\mu})$; we recall that, for $i=0,\dots,h-1$, each $\nu_i$  has poles $s_\mu$ and $t_\mu$.
A first condition for pair $(\pi_{s_\mu}, \pi_{t_\mu})$ to be a compatible pair for $\mu$ is that $(\pi_{s_\mu}, \pi_{t_\mu})$ is also a compatible pair for $\nu_i$, with $i=0,\dots,h-1$.
A second condition asks that the pair $(\pi_{s_\mu}, \pi_{t_\mu})$ \emph{defines opposite orders on the poles of $\mu$}.
Namely, let $W^s_\mu$ (resp., $W^t_\mu$) be the wheel of $V_{s_\mu}$ (resp., $V_{t_\mu}$) consistent with $\pi_{s_\mu}$ (resp., $\pi_{t_\mu}$).
Traversing clockwise the external cycle of $W^s_\mu$ starting from the first vertex of $ESeq(\mu,s_\mu,\pi_{s_\mu})$, let $ISeq(\nu_0,s_\mu,\pi_{s_\mu})$, $ISeq(\nu_1,s_\mu,\pi_{s_\mu})$, \dots, $ISeq(\nu_{h-1},s_\mu,\pi_{s_\mu})$ be the order by which the internal sequences are encountered.
Pair $(\pi_{s_\mu}, \pi_{t_\mu})$ defines opposite orders on the poles of $\mu$ if, traversing clockwise the external cycle of $W^t_\mu$ starting from the first vertex of $ESeq(\mu,t_\mu,\pi_{s_\mu})$, the order by which we encounter the internal sequences of $\nu_0,\nu_1,\dots, \nu_{h-1}$ is the opposite one, i.e., the order is $ISeq(\nu_{h-1},t_\mu,\pi_{t_\mu})$, $ISeq(\nu_{h-2},t_\mu,\pi_{t_\mu})$, \dots, $ISeq(\nu_0,t_\mu,\pi_{t_\mu})$.

\smallskip\noindent
\textbf{Node $\mu$ is an S-node.}
Let $\nu_0$ and $\nu_1$ be the children of $\mu$ such that $s_{\nu_0}=s_\mu$, $t_{\nu_0}=s_{\nu_1}$, and $t_{\nu_1}=t_\mu$.
A pair $(\pi_{s_\mu}, \pi_{t_\mu})$ is a compatible pair for $\mu$ if there exists a permutation $\pi_{t_{\nu_0}}$ such that the pair $(\pi_{s_\mu}, \pi_{t_{\nu_0}})$ is compatible for $\nu_0$ and the pair $(\pi_{t_{\nu_0}}, \pi_{t_\mu})$ is compatible for $\nu_1$.

Fig.~\ref{fi:p-wheels} suggests that a \nt planar representation of a clustered graph $G$ defines a permutation assignment $\Pi$ such that, for every node $\mu$ of $T$, pair $(\pi_{s_\mu},\pi_{t_\mu})$ is a compatible pair for $\mu$.

\begin{lemma}\label{le:compatible-pairs}
Let $G=(V,E,\mathcal{C},\Phi)$ be a clustered graph with side assignment $\Phi$ and let $T$ be the SPQ decomposition tree of the frame graph of $G$. Graph $G$ is \nt planar if and only if there exists a permutation assignment $\Pi$ such that, for every node $\mu$ of $T$ with poles $s_\mu$ and $t_\mu$, we have that permutation $\pi_{s_\mu} \in \Pi$ and permutation $\pi_{t_\mu} \in \Pi$ form a compatible pair of permutations for~$\mu$.
\end{lemma}
\begin{proof}
	We prove first that, if $G=(V,E,\mathcal{C},\Phi)$ is \nt planar, then there exists a permutation assignment $\Pi$ such that, for every node $\mu$ of $T$ with poles $s_\mu$ and $t_\mu$, the pair $(\pi_{s_\mu}, \pi_{t_\mu})$ is compatible for~$\mu$.
	
	Let $\Gamma$ be a \nt planar representation of $G$ with side assignment $\Phi$ and let $M_0, M_1, \dots, M_{h-1}$ be the matrices representing the non-trivial clusters of $G$.
	For each matrix $M_i$ ($i=0,\dots,h-1$) of $\Gamma$, let $\pi_i=v_0, v_1, \dots, v_{k-1}$ be the left to right order of the columns of $M_i$.
	We replace $M_i$ with a wheel $W_i$ consisting of a vertex $w_i$ of degree $4k$ adjacent to all vertices of a cycle $v_{0,\textrm{\sc{t}}}, v_{1,\textrm{\sc{t}}}, \dots, v_{k-1,\textrm{\sc{t}}}$, $v_{0,\textrm{\sc{r}}}, v_{1,\textrm{\sc{r}}}, \dots, v_{k-1,\textrm{\sc{r}}}$, $v_{k-1,\textrm{\sc{b}}}, v_{k-2,\textrm{\sc{b}}}, \dots, v_{0,\textrm{\sc{b}}}$, $v_{k-1,\textrm{\sc{l}}}$, $v_{k-2,\textrm{\sc{l}}}$, $\dots, v_{0,\textrm{\sc{l}}}$.
	For all $j= 0, 1, \dots, k-1$ and for all $\textrm{\sc{x}} \in \{ \textrm{\sc{t}}, \textrm{\sc{b}}, \textrm{\sc{l}}, \textrm{\sc{r}} \}$, vertex $v_{j,\textrm{\sc{x}}}$ is drawn at the point $p_{v_j,\textrm{\sc{x}}}$, that represents the attachment of the inter-cluster edges incident to vertex $v_j$ on the side $\textrm{\sc{x}}$ of matrix $M_i$. The edges of the external cycle of $W_i$ are drawn along the external boundary of $M_i$.
	Every inter-cluster edge $e=(u,v_j)$, with $v_j \in M_i$, is incident to the vertex $v_{j,\phi_i(e)}$ of the wheel $W_i$. Also, for all $j= 0, 1, \dots, k-1$ and for all $\textrm{\sc{x}} \in \{ \textrm{\sc{t}}, \textrm{\sc{b}}, \textrm{\sc{l}}, \textrm{\sc{r}} \}$, the cyclic order of the inter-cluster edges incident to $v_{j,\textrm{\sc{x}}}$ in $W_i$ is the same as the cyclic order of the inter-cluster edges incident to $p_{v_j,\textrm{\sc{x}}}$ in $M_i$.
	It is straightforward to verify that the computed drawing defines a planar embedding for the wheel reduction of $G=(V,E,\mathcal{C},\Phi)$ consistent with $\Pi=\{\pi_1, \pi_2, \dots, \pi_{|\mathcal{C}|}\}$. From the planarity of the wheel reduction of $G$ it follows that each non-trivial pole $v_\mu$ of the frame graph $F$ has a complete internal and a complete external sequence consistent with $\pi_{v_\mu}$ and that for every node $\mu$ of the SPQ decomposition tree of $F$
	having poles $s_\mu$ and $t_\mu$, the pair $(\pi_{s_\mu},\pi_{t_\mu})$, with $\pi_{s_\mu},\pi_{t_\mu} \in \Pi$, is compatible for~$\mu$. An example of the above described procedure is illustrated in Fig.~\ref{fi:p-wheels}.
	
	We now show that, if there exists a permutation assignment $\Pi$ such that for every node $\mu$ of $T$ with poles $s_\mu$ and $t_\mu$ we have that permutation pair $(\pi_{s_\mu},\pi_{t_\mu})$ is compatible for~$\mu$, then $G=(V,E,\mathcal{C},\Phi)$ is \nt planar with side assignment $\Phi$. We construct a planar embedding of the wheel reduction of $G$ consistent with $\Pi$ such that all external cycles of the wheels are embedded clockwise which, by Theorem~\ref{th:wheel-reduction}, implies that $G$ is \nt planar.
	Let $W_{s_\mu}$ and $W_{t_\mu}$ be the two wheels consistent with $\pi_{s_\mu}$ and $\pi_{t_\mu}$ of $s_\mu$ and $t_\mu$, respectively. We visit $T$ from the leaves to the root and incrementally construct the desired planar embedding of the wheel reduction of~$G$.
	
	If the visited node $\mu$ is a Q-node, at most one of its poles is non-trivial because $G$ is light; assume, without loss of generality, that the non-trivial pole of $\mu$ is $s_\mu$ and let $V_{s_\mu}$ be the cluster of $G$ represented by $s_\mu$ in the frame graph of $G$. We embed the wheel $W_{s_\mu}$ of $V_{s_\mu}$ consistent with $\pi_{s_\mu} \in \Pi$ such that, when traversing the edges of the external cycle of $W_{s_\mu}$ in the forward direction, the cycle is traversed clockwise. We embed $t_\mu$ in the external face of $W_{s_\mu}$ and planarly connect the top, bottom, left, or right copy of its end-vertex on $W_{s_\mu}$ as specified by $\Phi$.
	
	Suppose now $\mu$ is an S-node and let $\nu_0$ and $\nu_1$ be the children of $\mu$ such that $s_{\nu_0}=s_\mu$, $t_{\nu_0}=s_{\nu_1}$, and $t_{\nu_1}=t_\mu$. The planar embedding of the wheel reduction at node $\mu$ is obtained by composing the planar embedding of the wheel reduction at node $\nu_0$ with the planar embedding of the wheel reduction at $\nu_1$. This is done by identifying the planar embedding of the wheel $W_{t_{\nu_0}}$ consistent with $\pi_{t_{\nu_0}}$ with the planar embedding of the wheel $W_{s_{\nu_1}}$ consistent with $\pi_{s_{\nu_1}}$. Note that this is possible because $\pi_{t_{\nu_0}}$ is the same as $\pi_{s_{\nu_1}}$, $(\pi_{s_{\mu}},\pi_{t_{\nu_0}})$ is a compatible pair for $\nu_0$, and $(\pi_{s_{\nu_1}},\pi_{t_{\mu}})$ is a compatible pair for $\nu_1$.
	
	Finally, assume $\mu$ is a P-node and let $\nu_0, \nu_1, \dots, \nu_{h-1}$ be the children of $\mu$. Similarly to the case of the S-node, the planar embedding of the wheel reduction at node $\mu$ is obtained by composing the planar embeddings of the wheel reductions at nodes $\nu_0, \nu_1, \dots, \nu_{h-1}$. Since pair $(\pi_{s_\mu}, \pi_{t_\mu})$ is compatible for $\mu$, it defines opposite orders on the poles of $\mu$. These opposite circular orders correspond to a planar embedding of the wheel reduction at $\mu$ obtained by combining the planarly embedded wheel reductions at its children $\nu_0, \nu_1, \dots \nu_{h-1}$.
	It follows that $G=(V,E,\mathcal{C},\Phi)$ is \nt planar with permutation assignment~$\Pi$.\qed
\end{proof}

\begin{lemma}\label{le:series-parallel}
Let $G=(V,E,\mathcal{C},\Phi)$ be a flat clustered series-parallel graph with side assignment $\Phi$. Let $k$ be the maximum size of any cluster in $\mathcal{C}$ and let $n$ be the cardinality of $V$. There exists an
$O(k^{3k+\frac{3}{2}} \cdot n)$-time
algorithm that tests whether $G$ is \nt planar with side assignment $\Phi$ and if so, it computes a \nt planar representation of $G$ consistent with $\Phi$.
\end{lemma}
\begin{proof}
Let $F$ be the frame graph of $G$. 
We construct the SPQ decomposition tree $T$ of $G$ rooted at an arbitrary Q-node.
We visit $T$ from the leaves to the root and test whether $G$ has a permutation assignment $\Pi$ such that $G=(V,E,\mathcal{C},\Phi,\Pi)$ is \nt planar.
We first equip each non-trivial pole $v_\mu$ of every node $\mu$ of $T$ with its possible complete internal and complete external sequences.
The maximum number of complete internal sequences of $v_\mu$ is $k!$. The same is true for the complete external sequences. 
Each complete (internal or external) sequence of pole $v_\mu$ is encoded by means of the permutation $\pi_{v_\mu}$ and by the first and last vertex of the sequence in the clockwise order around $W_{v_\mu}$.
It follows that the intersection or the union of two complete internal or external sequences of the same permutation $\pi_{v_\mu}$ can be computed in constant time. Therefore, all complete internal and external sequences for each non-trivial pole of $T$ can be computed in $O(k!)$ time. Hence, the whole bottom-up traversal to equip all non-trivial poles with every possible complete internal/external sequence can be executed in $O(k!\cdot n)$ time.
We now test whether there exists a permutation assignment $\Pi$ such that any node $\mu$ of $T$ has a compatible pair of permutations. To this aim, we look at the complete internal and external sequences for the pair of poles of the children of $\mu$. For each pair $(\pi_{s_\mu}, \pi_{t_\mu})$ of permutations of the poles of $\mu$ we equip $\mu$ with the information about whether such a pair is compatible for $\mu$. This requires $O(k!^2)$ space.

If $\mu$ is a Q-node, every pair of permutations $(\pi_{s_\mu},\pi_{t_\mu})$ is compatible for $\mu$. It follows that all compatible pairs for $\mu$ can be computed in $O(k!)$ time (recall that one between $s_{\mu}$ and $t_{\mu}$ is non-trivial) and, hence, in $O(k! \cdot n)$ time for all the Q-nodes of $T$.

If $\mu$ is a P-node with children $\nu_0, \nu_1, \dots, \nu_{h-1}$, $\pi_{s_\mu}$ is one of the permutations that equip $s_\mu$, and $\pi_{t_\mu}$ is one of the permutations that equip $t_\mu$, testing whether the pair $(\pi_{s_\mu}, \pi_{t_\mu})$ is a compatible pair for $\mu$ can be executed in $O(h)$ time. It follows that all compatible pairs for $\mu$ can be computed in $O(k!^2 \cdot h)$ time and, hence, in $O(k!^2 \cdot n)$ time for all P-nodes of $T$.

If $\mu$ is an S-node with children $\nu_0$ and $\nu_1$, $\pi_{s_\mu}$ is one of the permutations that equip $s_\mu$, and $\pi_{t_\mu}$ is one of the permutations that equip $t_\mu$, testing whether the pair $(\pi_{s_\mu}, \pi_{t_\mu})$ is a compatible pair for $\mu$ can be executed in $O(k!)$ time, corresponding to choosing all possible permutations for the pole shared between $\nu_0$ and $\nu_1$. It follows that all compatible pairs for $\mu$ can be computed in $O(k!^3)$ time and, hence, in $O(k!^3 \cdot n)$ time for all S-nodes of $T$. 

Hence, the overall cost of the above described algorithm is $O(k!^3 \cdot n)$. 
It remains to prove that $O(k!^3 \cdot n) = O(k^{3k+\frac{3}{2}} \cdot n)$. 
By Stirling's approximation, $k! \sim \sqrt{2\pi k}(\frac{k}{e})^k$ and thus a series-parallel clustered graph $G$ with $n$ vertices, side assignment $\Phi$, and maximum cluster size $k$ can be tested for \nt planarity in $O(k^{3k+\frac{3}{2}} \cdot n)$ time. Note that the compatible pair of permutations stored at each node $\mu$ of $T$ implicitly define a planar embedding of a wheel reduction of $G$. It can be shown that it is possible to construct a \nt planar representation of~$G$ in time proportional to the number of edges of $G$, which is $O(n\cdot k)$~\cite{dlpt-pkng-17}. The statement of the lemma follows.\qed
\end{proof}

%
%
\subsection{Partial 2-Trees}\label{sse:partial-2-trees}

We now consider clustered graphs whose cluster size is at most $k$ and such that their frame graph is a partial $2$-tree, i.e., it is a planar graph whose biconnected components are series-parallel.
We handle this case by decomposing the frame graph into its blocks and we store them into a block-cut-vertex tree.
The following theorem generalizes the result of Lemma~\ref{le:series-parallel}.

\begin{theorem}\label{th:partial-2-tree}
Let $G=(V,E,\mathcal{C},\Phi)$ be a flat clustered partial $2$-tree with side assignment $\Phi$. Let $k$ be the maximum size of any cluster in $\mathcal{C}$ and let $n$ be the cardinality of $V$. There exists an $O(k^{3k+\frac{3}{2}}\cdot n)$-time algorithm that tests whether $G$ is \nt planar with side assignment $\Phi$ and if so, it computes a \nt planar representation of $G$ consistent with $\Phi$.
\end{theorem}
\begin{proof}
	We compute a block-cut-vertex tree $T_\textsc{bcv}$ of the frame graph of $G$, we root it at a block $B_{root}$ and perform a post-order traversal of $T_\textsc{bcv}$.
	Let $B_i$ be the currently visited block and let $c$ be the parent cut-vertex of $B_i$ in $T_\textsc{bcv}$. 
	We execute the testing algorithm of Lemma~\ref{le:series-parallel} by rooting the SPQ decomposition tree of $B_i$ at an arbitrary Q-node.
	If the test fails at any block of $T_\textsc{bcv}$, we conclude that $G$ is not \nt planar with the given side assignment. 
	Otherwise, we test whether, among the permutation assignments computed for the blocks $B_0, B_1, \dots, B_{h-1}$ that are children of a same cut-vertex $c$, there exists a set $\{\Pi_0, \Pi_1, \dots, \Pi_{h-1}\}$ such that: (i) $\pi_{0,c} = \pi_{1,c}= \dots = \pi_{h-1,c}$ with $\pi_{j,c} \in \Pi_j$, for $j=0,1, \dots, h-1$, and (ii) the complete internal sequence of $c$ for the block $B_j$ and for the permutation assignment $\pi_{j,c}$ does not overlap with the complete internal sequence of $c$ for the block $B_i$ and the permutation assignment $\pi_{i,c}$, with $i \neq j$, $i,j=0,\dots,h-1$. We equip $c$ with all the permutations that pass this test. Let $B'$ be the block that is the parent of $c$ in $T_\textsc{bcv}$. When testing $B'$ for \nt planarity, we consider for $c$ only the permutations that have been computed when processing blocks $B_0, B_1, \dots, B_{h-1}$ and check that the complete internal sequences of $c$ in $B_j$, $j=0,\dots,h-1$, do not intersect with the complete internal sequences of $c$ in $B'$.
	Let $n_B$ be the number of vertices of a block $B$. By using Lemma~\ref{le:series-parallel}, the procedure described above can be executed in time $O(k^{3k+\frac{3}{2}} \cdot n_B)$ for block $B$. Therefore, the post-order traversal of $T_\textsc{bcv}$ can be computed in $O(k^{3k+\frac{3}{2}} \cdot n)$ time. \qed
\end{proof}

%
%
\section{General Planar Frame Graphs}\label{se:general-planar}

In this section we study the problem of extending Theorem~\ref{th:partial-2-tree} to planar frame graphs that may not be partial 2-trees. We prove that \nt planarity testing with fixed sides can be solved in linear time for maximum cluster size $k=2$ (Subsection~\ref{sse:polynomial}). However, the problem becomes NP-complete with fixed sides for $k>2$ (Subsection~\ref{sse:np-completeness-fixed-sides}).

%
%
\subsection{Linearity for $k=2$}\label{sse:polynomial}

In~\cite{dlpt-nptsc-17-gd} we showed that \nt planarity testing with fixed sides is polynomial for maximum cluster size $k=2$. However, the proof is rather long and the proposed algorithm is cubic. Here we show a simpler linear algorithm that solves the same problem.

In order to prove linearity for \nt planarity with fixed sides when $k=2$, we extend the set of planarity constraints of~\cite{gkm-ptoei-08} with a new constraint that we call ``synchronized mirror constraint''. According to~\cite{gkm-ptoei-08}, given a graph $G=(V,E)$, an embedding constraint at a vertex $v \in V$ is a rooted, ordered tree $T_v$ such that its leaves are exactly the edges incident to $v$. The inner nodes of $T_v$ are of three types: \emph{oc-nodes} (stands for ``oriented constraint-nodes'') whose children have fixed clockwise order; \emph{mc-nodes} (stands for ``mirror constraint-nodes'') whose children have a fixed order up to a flip; and \emph{gc-nodes} (stands for ``grouping constraint-nodes''), whose children may be permuted. Since $T_v$ is an ordered tree, it imposes an order on its leaves and thus on the edges incident to~$v$. Tree $T_v$ for a vertex $v$ without any constraint has only the root, which is a gc-node.

In~\cite{gkm-ptoei-08} it is shown how these embedding constraints can be handled in linear time by replacing each vertex with a gadget representing the embedding constraint $T_v$. Namely, each vertex $v$ is replaced by a construction reproducing the nodes and edges of tree $T_v$, where each oc-node and mc-node is replaced by a wheel graph and each gc-node is replaced by a regular node. The obtained auxiliary graph, that we denote by $H_G$, can be constructed in time $O(n)$, where $n$ is the number of vertices of $G$~\cite[Lemma 1]{gkm-ptoei-08}.
Graph $H_G$ is tested for planarity. If $H_G$ is planar, a planar embedding is constructed such that for every wheel graph corresponding to a constraint, its external cycle does not contain any vertex other than the center of the wheel in its interior.
The algorithm terminates by checking whether the rigid components of the SPQR decomposition tree of the biconnected components of $H_G$ can be flipped in such a way that all the wheels representing oc-nodes have the same clockwise order as specified by the corresponding constraints. From the obtained planar embedding of $H_G$, by contracting each tree to a vertex, a planar embedding of $G$ satisfying the input constraints can be constructed (see also~\cite{gkm-ptoei-08} for more details).

We extend the approach of~\cite{gkm-ptoei-08} by introducing a fourth type of node in the definition of tree $T_v$. Namely, an \emph{smc-node} (which stands for ``synchronized mirror constraint-node'') has, in addition to a circular order of its children, also a color (an integer). A synchronized mirror constraint is satisfied if the children of all smc-nodes with the same color all appear in the clockwise order fixed by the constraints or if they all appear in counter-clockwise order. Intuitively, smc-nodes generalize mc-nodes, allowing to reverse the order of the neighbors of some vertices in a synchronized way.
If an smc-node has a color that is not shared with other smc-nodes, then the constraint represented by such a node is equivalent to that represented by an mc-node.
Fig.~\ref{fi:mutzel-plus-plus} shows an example of a constraint tree $T_v$ containing two smc-nodes. Fig.~\ref{fi:mutzel-plus-plus-1} shows a vertex $v$ and its incident edges. Fig.~\ref{fi:mutzel-plus-plus-2} represents a set of constraints for the edges incident to $v$. Fig.~\ref{fi:mutzel-plus-plus-3} shows the replacement of the constraints with suitable gadgets in the auxiliary graph $H_G$.

\begin{figure}[htb]
\centering
\subfigure[]{\includegraphics[page=1,width=0.22\textwidth]{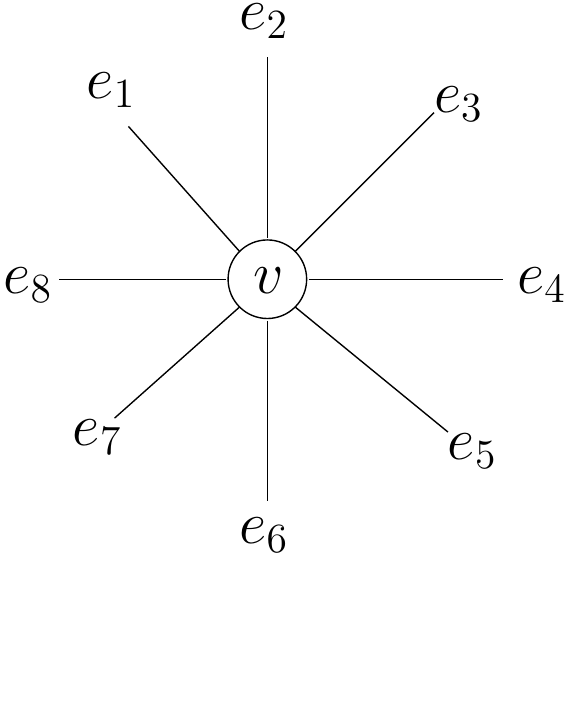}\label{fi:mutzel-plus-plus-1}}
\hfil
\subfigure[]{\includegraphics[page=2,width=0.33\textwidth]{mutzel-plus-plus.pdf}\label{fi:mutzel-plus-plus-2}}
\hfil
\subfigure[]{\includegraphics[page=3,width=0.37\textwidth]{mutzel-plus-plus.pdf}\label{fi:mutzel-plus-plus-3}}
\caption{An example to illustrate the constrained embedding of the smc-node: (a) A vertex $v$ and its incident edges. (b) The embedding constraints for $v$ described by its constraint tree $T_v$ with oc-, mc-, gc-nodes, and two smc-nodes with the same color $x$. (c) The replacement of the constraints with suitable wheels.}\label{fi:mutzel-plus-plus}
\end{figure}

Fig.~\ref{fi:mutzel-plus-plus-example-1} shows an example of a planar graph that does not admit a planar embedding if the clockwise order of the edges incident to $v_1$ is ``synchronized'' with the clockwise order of the edges incident to~$v_2$, as described by Fig.~\ref{fi:mutzel-plus-plus-example-2}. More precisely, the graph has no planar embedding if one requires that either the clockwise order of the vertices adjacent to $v_1$ is $v_3, v_2, v_4$ and the clockwise order of the vertices adjacent to $v_2$ is $v_1, v_4, v_3$ or that both these clockwise orders are reversed, as depicted in Figs.~\ref{fi:mutzel-plus-plus-example-3} and~\ref{fi:mutzel-plus-plus-example-4}.

\begin{figure}[htb]
\centering
\subfigure[]{\includegraphics[page=1,width=0.18\textwidth]{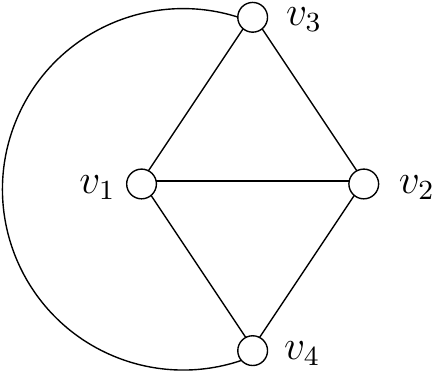}\label{fi:mutzel-plus-plus-example-1}}
\hfil
\subfigure[]{\includegraphics[page=2,width=0.25\textwidth]{mutzel-plus-plus-example.pdf}\label{fi:mutzel-plus-plus-example-2}}
\hfil
\subfigure[]{\includegraphics[page=3,width=0.20\textwidth]{mutzel-plus-plus-example.pdf}\label{fi:mutzel-plus-plus-example-3}}
\hfil
\subfigure[]{\includegraphics[page=4,width=0.20\textwidth]{mutzel-plus-plus-example.pdf}\label{fi:mutzel-plus-plus-example-4}}
\caption{An example of unsatisfiable planarity constraints involving smc-nodes: (a) A small graph $G$; (b) two synchronized embedding constraints on two vertices of the graph; (c) and (d) two non-planar drawings of the auxiliary graph $H_G$.}\label{fi:mutzel-plus-plus-example}
\end{figure}

We now show how to take into account the constraints represented by smc-nodes. In its more general setting the graph $G$ has vertices constrained by trees that contain oc-, mc-, gc-, and smc-nodes. Construct the auxiliary graph $H_G$ by first replacing each vertex $v$ of $G$ with its associated tree $T_v$ and then replacing each oc-node, mc-node, and smc-node of $T_v$ with a suitable wheel graph. In general, graph $H_G$ may have several connected components. Consider the block-cut-vertex tree for each such component. Note that each wheel subgraph of $H_G$ is entirely contained inside a single block. Each block can be associated with an SPQR-tree and each wheel subgraph is entirely contained in the skeleton of an R-node of some SPQR-tree (see Observation 1 of~\cite{gkm-ptoei-08}).

Let $\Gamma(H_G)$ be an arbitrary planar embedding of $H_G$ and let $\mu$ be an R-node of some SPQR-tree associated with a block of the block-cut-vertex tree of $H_G$.
We associate $\skel(\mu)$ with a Boolean variable $x_\mu$ and we say that, in an embedding $\Gamma'(H_G)$ of $H_G$, $x_\mu$ is \texttt{true} if $\skel(\mu)$ is embedded as in $\Gamma(H_G)$, while $x_\mu$ is \texttt{false} otherwise.

Let $\mu$ be an R-node whose skeleton $\skel(\mu)$ entirely contains the wheel graph corresponding to an oc-node. The presence of the oc-node imposes a constraint on the two possible embeddings of $\skel(\mu)$. Namely, only the embedding of $\skel(\mu)$ where the wheel graph is embedded clockwise satisfies the oc-node constraint. This implies a truth value (\texttt{true} or \texttt{false}) for the associated variable $x_\mu$. Note that, if two oc-nodes are contained in the same $\skel(\mu)$ and they assign opposite values to $x_\mu$, then $G$ does not admit an embedding that satisfies the given constraints.

Let $\mu$ be an R-node whose skeleton $\skel(\mu)$ entirely contains the wheel graph $W$ corresponding to an smc-node with color $c$. Consider any other skeleton $\skel(\mu')$ containing a wheel graph $W'$ corresponding to an smc-node with the same color $c$. If in $\Gamma(H_G)$ the wheels $W$ and $W'$ are both oriented either clockwise or counter-clockwise, then $x_\mu = x_{\mu'}$. Otherwise, $x_\mu = \overline{x}_{\mu'}$.

The above constraints easily translate to suitable 2SAT clauses. Namely, $x_\mu = x_{\mu'}$ can be expressed as $(x_\mu \vee \overline{x}_{\mu'}) \wedge (\overline{x}_\mu \vee x_{\mu'})$, while $x_\mu = \overline{x}_{\mu'}$ can be expressed as $(x_\mu \vee x_{\mu'}) \wedge (\overline{x}_\mu \vee \overline{x}_{\mu'})$. Therefore, any solution of the 2SAT instance corresponds to an embedding of $H_G$ that satisfies the given constraints.

The assignments that satisfy the 2SAT formula correspond bijectively to the planar embeddings of $H_G$ that satisfy the constraints on the ordering of the edges around the vertices expressed by the oc-, mc-, gc-, and smc-nodes. In particular, it suffices to have one arbitrary satisfying assignment to construct such an embedding. The embedding of graph $G$ is obtained by contracting all the edges of $H_G$ that are not edges of $G$, analogously to~\cite{gkm-ptoei-08}.
The discussion above implies the following.

\begin{theorem}\label{th:mutzel++}
Let $G$ be an $n$-vertex graph with embedding constraints $C$ modeled by oc-, mc-, gc-, and smc-nodes. Then, we can test in time $O(n)$ whether $G$ admits a planar embedding satisfying all  constraints in $C$. In the affirmative case, one such planar embedding can be computed in time $O(n)$.
\end{theorem}
\begin{proof}
Graph $H_G$ can be constructed in linear time from $G$ and $C$ as described above.
Connected, biconnected and triconnected components of $H_G$ can be computed in linear time~\cite{ht-aeagm-73,ht-dgtc-73}. In particular, the structure of the triconnected components of each block can be described by an SPQR decomposition tree of linear size~\cite{dt-olpt-96}. Let $\Gamma(H_G)$ be an arbitrary planar embedding of $H_G$.
We now use the technique of~\cite[Lemma 6]{br-drpse-15} in order to construct in linear time the 2SAT formula $\varphi$ that describes the constraints on the planar embedding of $H_G$. In particular, this technique allows us to identify in overall linear time, for each wheel $W$, the R-node $\mu_W$ that contains $W$ and the orientation of $W$ in the embedding $\Gamma(H_G)$. Since 2SAT satisfiability can be tested in linear time~\cite{apt-ltatt-79,eis-cttmc-75}, it follows that testing whether $H_G$ admits a planar embedding satisfying the constraints in $C$ can be done in linear time. A planar embedding $\Gamma'(H_G)$ of $H_G$ that satisfies the constraints in $C$ can be constructed from a truth assignment that satisfies $\varphi$ by orienting the skeletons of each R-node $\mu$ as specified by the value of $x_\mu$. Finally, the embedding of $G$ that satisfies the constraints in $C$ is found from $\Gamma'(H_G)$ in linear time by contracting all the edges of $H_G$ that are not edges of $G$. \qed
\end{proof}

\begin{figure}[htb]
\centering
\subfigure[]{\includegraphics[page=1,width=0.32\textwidth]{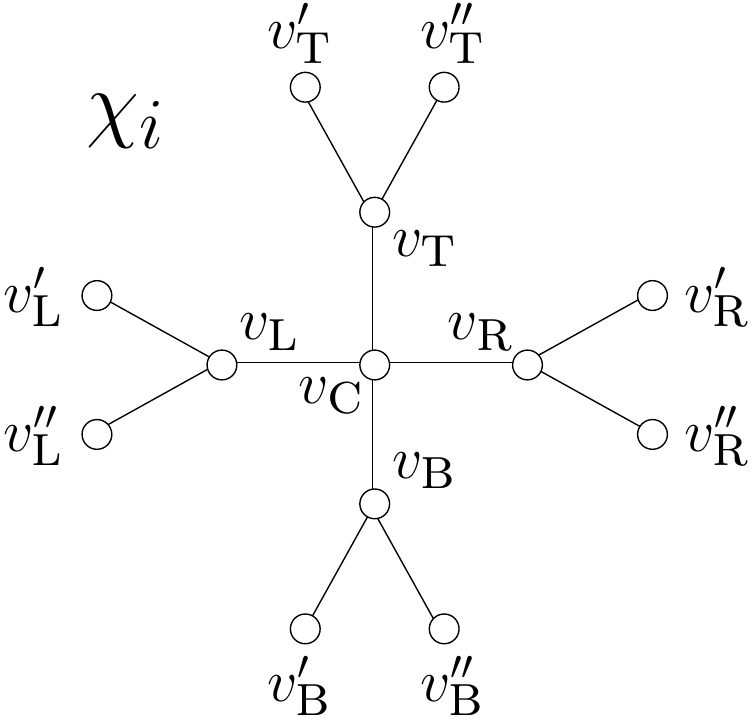}\label{fi:chi-gadget-1}}
\hfil
\subfigure[]{\includegraphics[page=2,width=0.33\textwidth]{ignaz.pdf}\label{fi:chi-gadget-2}}
\caption{The $\chi_i$ gadget used to represent a cluster $V_i = \{v',v''\}$. (a) The vertices and edges of $\chi_i$. (b) The embedding constraints for the vertices of  $\chi_i$.}\label{fi:chi-gadget}
\end{figure}

We now use Theorem~\ref{th:mutzel++} to prove that \nt planarity testing with fixed sides is linear for maximum cluster size $k=2$.
We construct an auxiliary graph $G'$ by replacing each cluster $V_i = \{v',v''\}$ of size~$2$ with a suitable gadget $\chi_i$ and a set of constraints associated with $\chi_i$. Gadget $\chi_i$ consists of a tree with $13$ vertices. Refer to Fig.~\ref{fi:chi-gadget}.
Tree $\chi_i$ has a vertex $v_\textrm{\sc{c}}$ of degree $4$ adjacent to vertices $v_\textrm{\sc{l}}, v_\textrm{\sc{r}}, v_\textrm{\sc{t}}$, and $v_\textrm{\sc{b}}$. Also, $v_\textrm{\sc{c}}$ has an embedding constraint represented by an oc-node that forces vertices $v_\textrm{\sc{t}}$, $v_\textrm{\sc{r}}$, $v_\textrm{\sc{b}}$, and $v_\textrm{\sc{l}}$ to appear in this clockwise order in the planar embedding of $\chi_i$. For $\textrm{\sc{x}} \in \{\textrm{\sc{l}},\textrm{\sc{r}},\textrm{\sc{t}},\textrm{\sc{b}}\}$, vertex $v_\textrm{\sc{x}}$ has degree $3$ and is adjacent to two leaves $v_\textrm{\sc{x}}'$, $v_\textrm{\sc{x}}''$.
For $\textrm{\sc{x}} \in \{\textrm{\sc{t}},\textrm{\sc{r}}\}$, vertex $v_\textrm{\sc{x}}$ has an embedding constraint represented by an smc-node with color $i$ that forces vertices $v_\textrm{\sc{c}}$, $v_\textrm{\sc{x}}'$, and $v_\textrm{\sc{x}}''$ to appear either in this clockwise order or in the reverse clockwise order (i.e., $v_\textrm{\sc{c}}$, $v_\textrm{\sc{x}}''$, and $v_\textrm{\sc{x}}'$) in any planar embedding of $\chi_i$.
For $\textrm{\sc{x}} \in \{\textrm{\sc{b}},\textrm{\sc{l}}\}$, vertex $v_\textrm{\sc{x}}$ has an embedding constraint represented by an smc-node with color $i$ that forces vertices $v_\textrm{\sc{c}}$, $v_\textrm{\sc{x}}''$, and $v_\textrm{\sc{x}}'$ to appear either in this clockwise order or in the reverse clockwise order (i.e., $v_\textrm{\sc{c}}$, $v_\textrm{\sc{x}}'$, and $v_\textrm{\sc{x}}''$) in any planar embedding of $\chi_i$.
Each inter-cluster edge $e$ incident to $v'$ (to $v''$, respectively) for which $\phi_i(e)=\textrm{\sc{x}}$ ($\textrm{\sc{x}} \in \{\textrm{\sc{l}},\textrm{\sc{r}},\textrm{\sc{t}},\textrm{\sc{b}}\}$) is incident to $v_\textrm{\sc{x}}'$ (to $v_\textrm{\sc{x}}''$, respectively).

Fig.~\ref{fi:ignaz-matrix-1} and Fig.~\ref{fi:ignaz-matrix-2} depict the two possible planar embeddings of $\chi_i$ satisfying the above described constraints.

\begin{figure}[htb]
\centering
\subfigure[]{\includegraphics[page=1,width=0.46\textwidth]{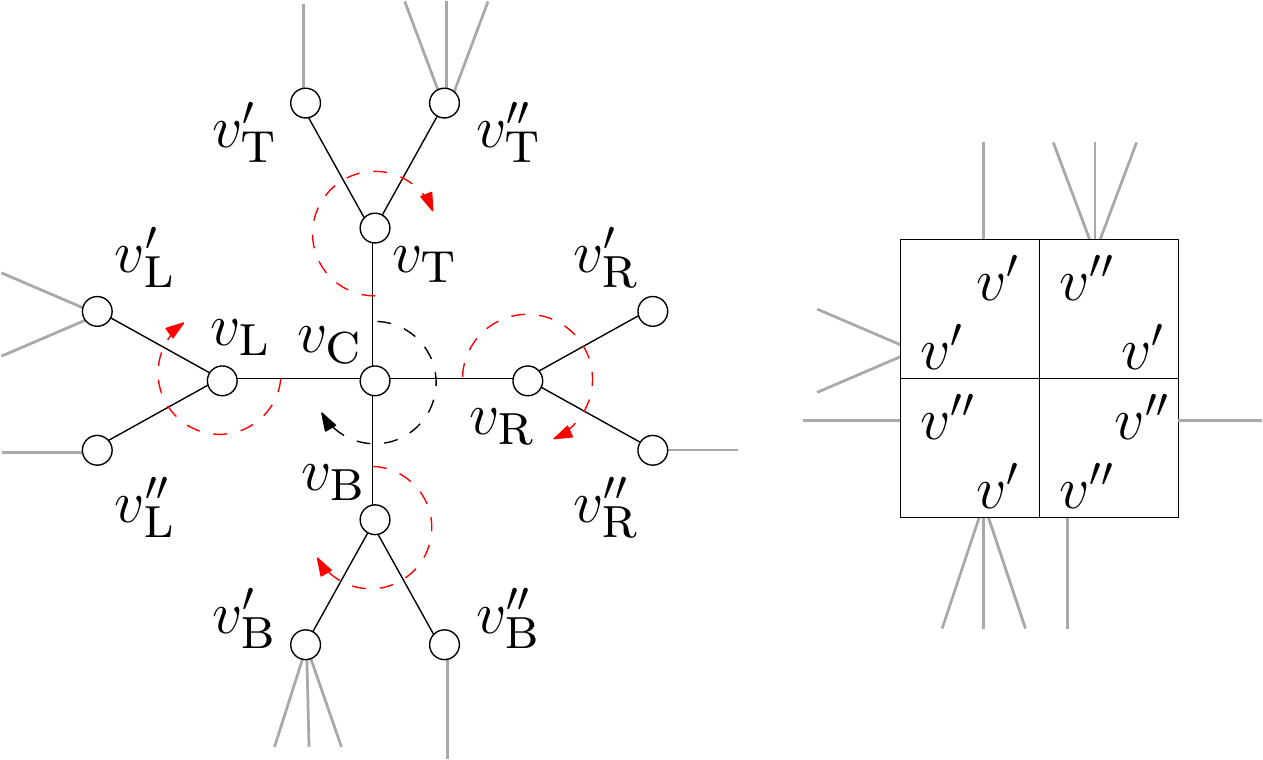}\label{fi:ignaz-matrix-1}}
\hfil
\subfigure[]{\includegraphics[page=2,width=0.46\textwidth]{ignaz-matrix.pdf}\label{fi:ignaz-matrix-2}}
\caption{The two possible embeddings of gadget $\chi_i$ and their effects on the circular order of the inter-cluster edges incident to $V_i = \{v',v''\}$. (a) The embedding of $\chi_i$ that corresponds to $\pi_i = v', v''$. (b) The embedding of $\chi_i$ that corresponds to $\pi_i = v'', v'$.}\label{fi:ignaz-matrix}
\end{figure}

Based on the discussion above we can prove the following theorem.

\begin{theorem}\label{th:polynomiality}
	Let $G=(V,E,\mathcal{C},\Phi)$ be an $n$-vertex clustered graph with side assignment $\Phi$ such that the maximum size of any cluster in $\mathcal{C}$ is two. There exists an $O(n)$-time algorithm that tests whether $G$ is \nt planar with the given side assignment and if so, it computes a \nt planar representation of $G$ consistent with $\Phi$.
\end{theorem}
\begin{proof}
	We construct a graph $G'$ by replacing each cluster $V_i \in \mathcal{C}$ with a gadget $\chi_i$ and its associated constraints as described above. Since every gadget $\chi_i$ has constant size, $G'$ has $O(n)$ vertices and edges.
	We test $G'$ for constrained planarity by means of the procedure in Theorem~\ref{th:mutzel++}. If the test is positive, let $\Gamma(G')$ be the planar embedding of $G'$ constructed by Theorem~\ref{th:mutzel++}. We construct an embedded graph $\Gamma(G'')$ from $\Gamma(G')$ as follows. Let $\chi_i$ be the gadget of $G'$ corresponding to cluster $V_i = \{v',v''\}$ of $G$. Assume that the clockwise order of the vertices adjacent to $v_\textrm{\sc{t}}$ in $\Gamma(G')$ is $v_\textrm{\sc{c}}$, $v_\textrm{\sc{t}}'$, and $v_\textrm{\sc{t}}''$ (if the order is opposite the proof is analogous).
	We replace $\chi_i$ in $\Gamma(G')$ with an embedded wheel graph $W_i$ centered at $v_\textrm{\sc{c}}$ and having as external boundary $v_\textrm{\sc{t}}'$, $v_\textrm{\sc{t}}''$, $v_\textrm{\sc{r}}'$, $v_\textrm{\sc{r}}''$, $v_\textrm{\sc{b}}''$, $v_\textrm{\sc{b}}'$, $v_\textrm{\sc{l}}''$, and $v_\textrm{\sc{l}}'$, appearing in this clockwise order around $v_\textrm{\sc{c}}$.
  	Since there is no inter-cluster edge incident to vertices $v_\textrm{\sc{c}}$, $v_\textrm{\sc{t}}$, $v_\textrm{\sc{r}}$, $v_\textrm{\sc{b}}$, and $v_\textrm{\sc{l}}$ of $\chi_i$, replacing $\chi_i$ with $W_i$ does not introduce any edge crossings. It follows that $\Gamma(G'')$ is a planar embedded graph and that it is a wheel reduction of $G$. Therefore, we can construct a \nt planar representation $\Gamma(G)$ of $G$ as in the proof of Theorem~\ref{th:wheel-reduction}.
  	To conclude the proof observe that both the construction of $\Gamma(G'')$ and of $\Gamma(G)$ can be executed in constant time per cluster, which implies that the overall complexity is linear. \qed
\end{proof}

%
%

\subsection{NP-Completeness for $k > 2$}\label{sse:np-completeness-fixed-sides}

The linear-time result of Theorem~\ref{th:polynomiality} can not be extended to the case when $k>2$. The following theorem extends the result of~\cite{ddfp-cnrcg-jgaa-17} and proves that \nt planarity testing with fixed sides is NP-complete even when the maximum size of the clusters is three. The proof is based on a reduction from (non-planar) NAE3SAT.

\begin{theorem}\label{th:fixed-sides-completeness}
\nt planarity testing with fixed sides and cluster size at most $k$ is NP-complete for any $k>2$.
\end{theorem}
\begin{proof}
	\nt planarity testing with fixed sides is trivially in NP. In fact, given a clustered graph $G=(V,E,\mathcal{C},\Phi)$, all possible permutations assignments $\Pi$ could be non-deterministically computed and the problem of deciding whether a clustered graph $G=(V,E,\mathcal{C},\Phi,\Pi)$ admits a \nt planar representation with side assignment $\Phi$ and permutation assignment $\Pi$ is solvable in linear time~\cite{ddfp-cnrcg-jgaa-17}.
	
	In order to prove its NP-hardness, we reduce NAE3SAT to it. An instance of NAE3SAT consists of a collection of clauses on a set of Boolean variables, where each clause consists of exactly three literals. The problem asks whether there exists a truth assignment to the variables so that each clause has at least one true literal and at least one false literal. This problem has been shown to be NP-complete by Thomas J. Schaefer~\cite{s-csp-78}. However, it is known to be polynomial when the graph of the adjacencies of the variables and clauses is planar~\cite{m-pnp-88}.
	
	Starting from a (non-planar) instance of NAE3SAT with variables $x_1$, $x_2$, \dots, $x_n$ and clauses $C_1, C_2, \dots, C_m$, we construct an instance $G=(V,E,\mathcal{C},\Phi)$ of \nt planarity testing with fixed sides as follows.
	First, we obtain a (non-planar) drawing $\Gamma$ of the graph of its variables and clauses like the one in Fig.~\ref{fi:nae3sat-instance}: The clause vertices are vertically aligned, the variable vertices are horizontally aligned, and each edge is drawn as an L-shape (i.e. an orthogonal polyline with one bend). Clearly this drawing can be computed in polynomial time and has a polynomial number of crossings.

	\begin{figure}[tb]
		\centering
		\includegraphics[width=7.2cm,page=2]{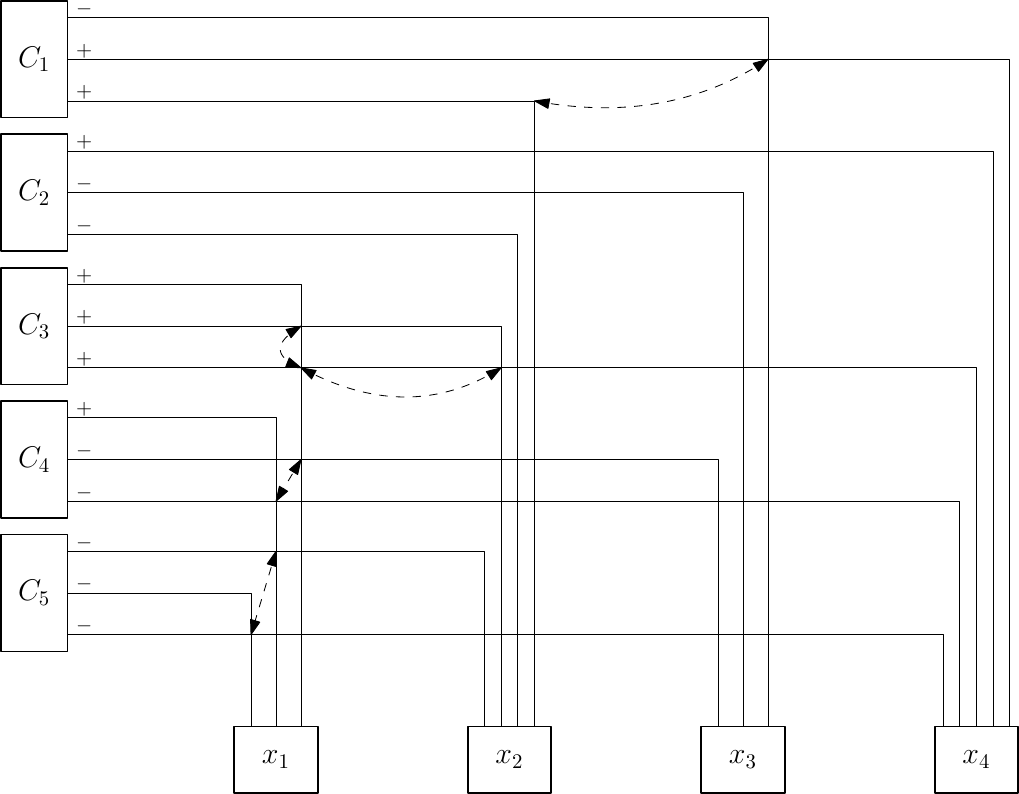}
		\caption{A non-planar drawing of an instance of NAE3SAT.
			The small plus and minus signs represent direct and negated occurrences of the variables in the clauses, respectively. 
		}\label{fi:nae3sat-instance}
	\end{figure}
	
	Then, we replace each vertex representing a variable with a variable gadget. The variable gadget for a variable of degree $h$ is composed by $h$ size-two clusters connected together as depicted in Fig.~\ref{fi:variable-gadget-true} and~\ref{fi:variable-gadget-false}. Namely, let $V_{v,1}, V_{v,2}, \dots, V_{v,h}$ be the $h$ clusters of $\mathcal{C}$ composing the variable gadget of variable $v$ and let $\{u_{i,1},u_{i,2}\}$ be the nodes of $V_{v,i}$, with $i=1, \dots, h$. We may encode a truth value with each one of the two possible representations of cluster $V_{i}$: If in the matrix $M_i$ representing $V_i$ the column corresponding to the vertex $u_{i,1}$ precedes the column corresponding to the vertex $u_{i,2}$, we say that $M_i$ is \texttt{true}. Otherwise, we say that $M_i$ is \texttt{false}. Correspondingly, we say that $\pi_i=u_{i,1},u_{i,2}$ is the \texttt{true} permutation of cluster $V_i$ and that $\pi_i=u_{i,2},u_{i,1}$ is its \texttt{false} permutation. In order to connect the clusters composing the variable gadget for $v$, we add to $E$,
	for $i=1,\dots,h-1$, the inter-cluster edges $e_{i,1}=(u_{i,1},u_{i+1,1})$ and $e_{i,2}=(u_{i,2},u_{i+1,2})$ and set $\phi_i(e_{i,1})=\textrm{\sc{r}}$, $\phi_{i+1}(e_{i,1})=\textrm{\sc{b}}$, $\phi_i(e_{i,2})=\textrm{\sc{r}}$, $\phi_{i+1}(e_{i,2})=\textrm{\sc{b}}$.
	It is immediate that in any \nt planar representation of $G$, all $M_i$, $i=1,\dots,h$, are either simultaneously \texttt{true} or simultaneously \texttt{false}. Correspondingly, we say that the variable gadget is \texttt{true} or \texttt{false}. Fig.~\ref{fi:variable-gadget-true} and~\ref{fi:variable-gadget-false} show an example of a \texttt{true} and of a \texttt{false} drawing of a variable gadget.
	
	\begin{figure}[tb]
		\centering
		\subfigure[]{\includegraphics[width=3cm]{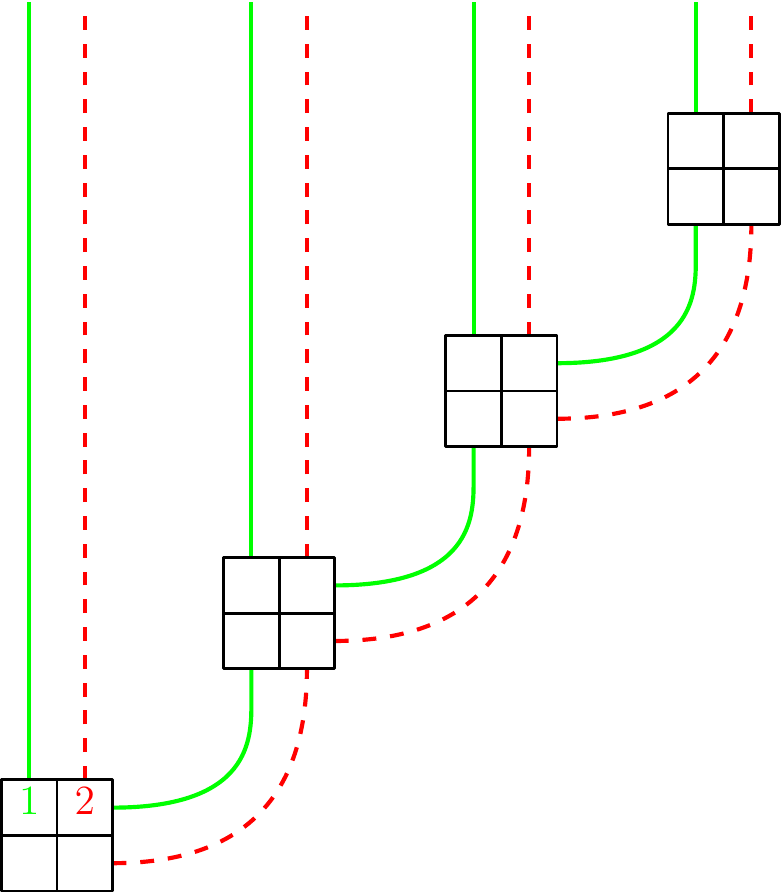}\label{fi:variable-gadget-true}}
		\hfill
		\subfigure[]{\includegraphics[width=3cm]{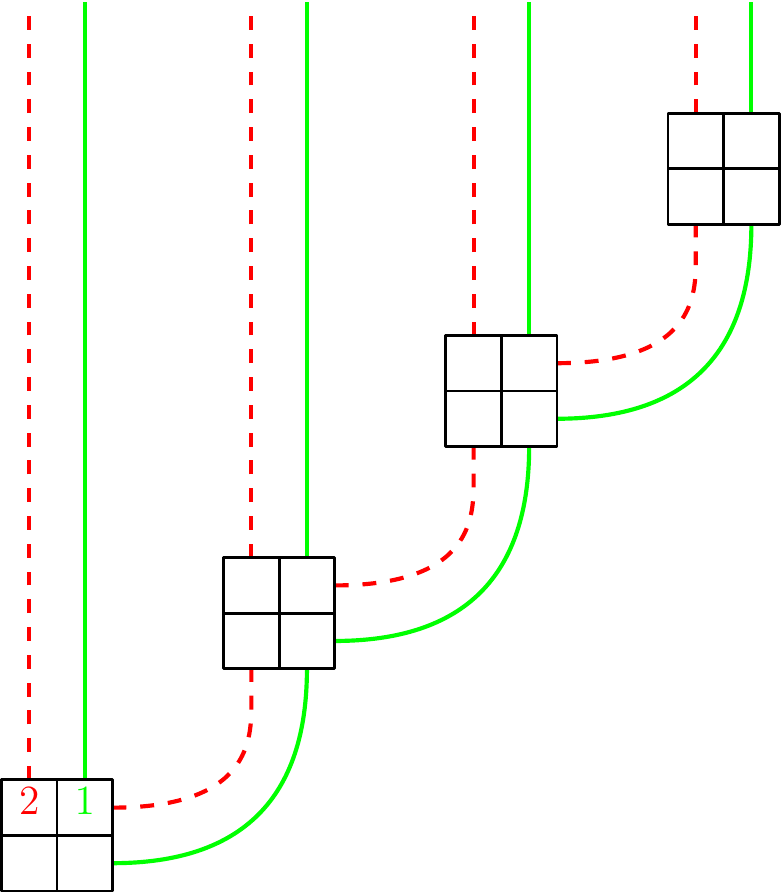}\label{fi:variable-gadget-false}}
		\hfill
		\subfigure[]{\includegraphics[width=1cm]{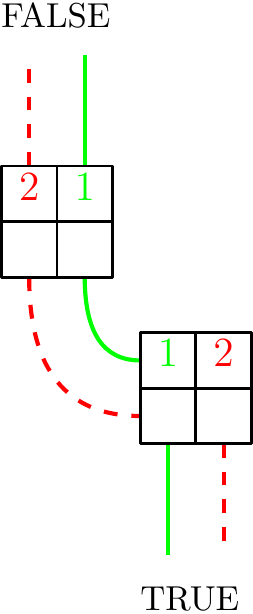}\label{fi:not-gadget-tf}}
		\hfill
		\subfigure[]{\includegraphics[width=1cm]{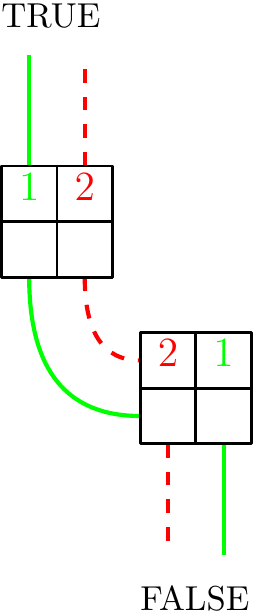}\label{fi:not-gadget-ft}}
		\caption{\subref{fi:variable-gadget-true} The \texttt{true} configuration of a variable gadget for a variable of degree four. \subref{fi:variable-gadget-false} The \texttt{false} configuration. \subref{fi:not-gadget-tf} The \texttt{not} gadget transforming an encoded \texttt{true} value into a \texttt{false} value. \subref{fi:not-gadget-ft} The \texttt{not} gadget transforming an encoded \texttt{false} value into a \texttt{true} value.}\label{fi:variable-and-not-gadgets}
	\end{figure}
	
	Each edge $e$ attaching to a variable $v$ in the drawing $\Gamma$ (refer to Fig.~\ref{fi:nae3sat-instance}) corresponds to two `parallel' inter-cluster edges $e_1$ and $e_2$ attached to one of the clusters composing the variable gadget of $v$. Let $V_j$ be such a cluster. We set $\phi_j(e_{1})=\textrm{\sc{t}}$ and $\phi_j(e_{2})=\textrm{\sc{t}}$. Observe that the order in which $e_1$ and $e_2$ exit $M_j$ depends on the truth value encoded by $M_j$ and, hence, on the truth value encoded by the variable gadget for $v$.
	
	\begin{figure}[tb]
		\centering
		\subfigure[]{\includegraphics[width=2cm]{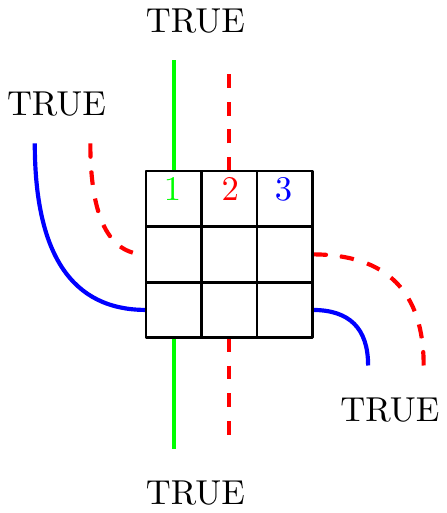}\label{fi:crossing-gadget-1}}
		\hfil
		\subfigure[]{\includegraphics[width=2cm]{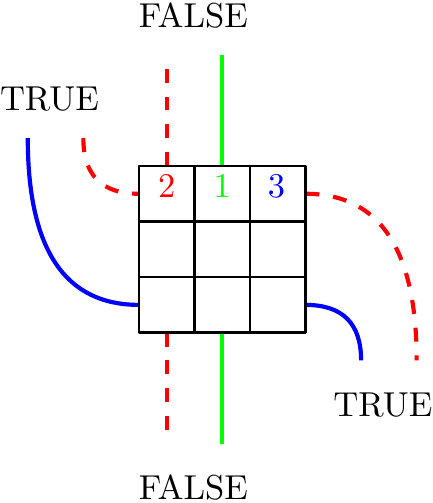}\label{fi:crossing-gadget-2}}
		\hfil
		\subfigure[]{\includegraphics[width=2cm]{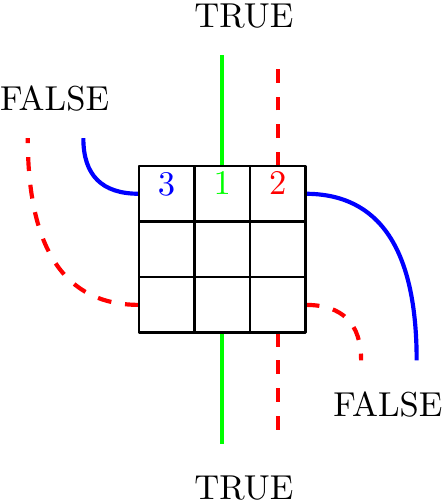}\label{fi:crossing-gadget-3}}
		\hfil \\
		\subfigure[]{\includegraphics[width=2cm]{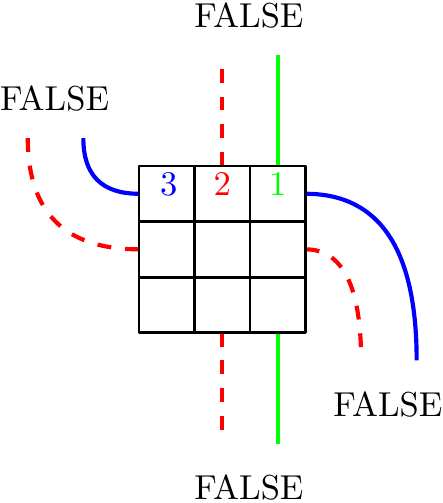}\label{fi:crossing-gadget-4}}
		\hfil
		\subfigure[]{\includegraphics[width=2cm]{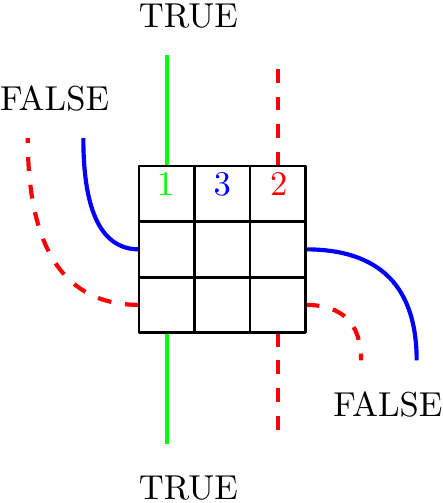}\label{fi:crossing-gadget-5}}
		\hfil
		\subfigure[]{\includegraphics[width=2cm]{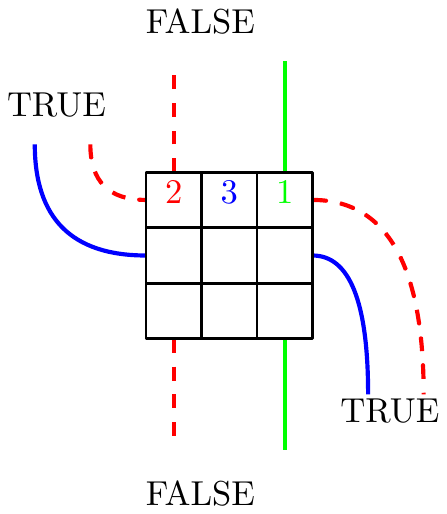}\label{fi:crossing-gadget-6}}
		\caption{The six possible configurations of a crossing gadget. }\label{fi:crossing-gadget}
	\end{figure}
	
	Fig.~\ref{fi:crossing-gadget} depicts the gadget we use to replace crossings in $\Gamma$, consisting of a cluster $V_x$ of size three. From the figure it is apparent that, in any representation of $V_x$, the left-to-right order of the edges entering $M_x$ from the bottom side is the same as the left-to-right order of the edges exiting $M_x$ from the top side. An analogous consideration holds for the top-to-bottom order of the edges entering $M_x$ from the right side and the top-to-bottom order of the edges exiting $M_x$ from the left side. This implies that the truth value encoded by the edges entering $M_x$ is the same as the truth value encoded by the edges exiting it.

	We now describe the clause gadget. We assume that three pairs of edges, encoding the truth value of the variables occurring in the clause, arrive to the clause gadget. Let $v_1, v_2,$ and $v_3$ be the three variables whose literals $l_1$, $l_2$, and $l_3$ occur in clause $C$. Before entering the clause gadget, if literal $l_i$ is a directed literal (i.e., if $l_i=v_i$) we introduce a size-two cluster that receives the truth value from $v_i$ from the bottom and transmits it to the clause gadget from the top. Otherwise, if literal $l_i$ is a negated literal of variable $v$ (i.e., if $l_i=\overline{v}_i$) then we attach the edges coming from $v_i$ to a \texttt{not} gadget, depicted in Fig.~\ref{fi:not-gadget-tf} and~\ref{fi:not-gadget-ft}, and use the edges exiting the \texttt{not} gadget instead of the edge coming directly from variable $v_i$. This has the effect that all the three pairs of edges entering the clause gadget encode a truth value that is \texttt{true} if the literal is \texttt{true} and \texttt{false} if the literal is \texttt{false}. In the following, therefore, we will consider the truth values of the literals, rather than the truth values of the variables.
	
	\begin{figure}[tb]
		\centering
		\includegraphics[width=6cm]{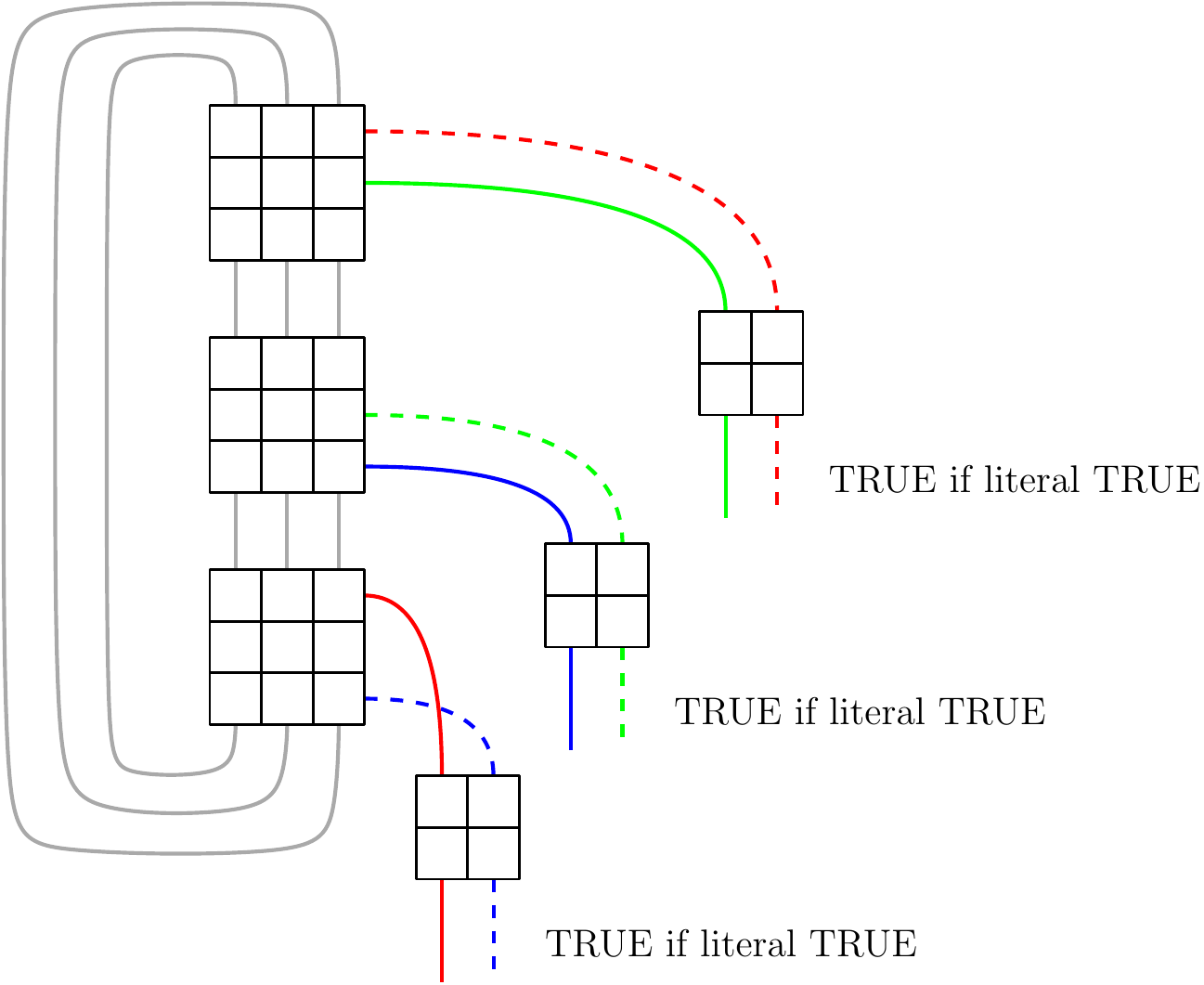}
		\caption{The clause gadget.}\label{fi:clause-gadget}
	\end{figure}
	
	The clause gadget, depicted in Fig.~\ref{fi:clause-gadget}, is composed by three clusters $V_1, V_2,$ and $V_3$ of size three, each having vertices $\{u_{i,a},u_{i,b},u_{i,c},\}$, $i=1,2,3$. The three clusters are connected together in such a way that, in any \nt planar representation of $G$, their permutations $\pi_i$, $i=1,2,3$, always present the same sequence of the labels $a,b,c$. For example, if $\pi_1 = u_{1,c},u_{1,a},u_{1,b}$, then also $\pi_2 = u_{2,c},u_{2,a},u_{2,b}$ and $\pi_3 = u_{3,c},u_{3,a},u_{3,b}$.
	
	For $i=1,\dots,3$, the edges encoding the truth value of literal $l_i$ attach to cluster $V_i$, where the prescribed side is the right side of matrix $M_i$. The two edges $e_{1,1}$ and $e_{1,2}$ encoding the truth value of the literal $l_1$ attach to $u_{1,a}$ and $u_{1,b}$, respectively. The two edges $e_{2,1}$ and $e_{2,2}$ encoding the truth value of the literal $l_2$ attach to $u_{2,b}$ and $u_{2,c}$, respectively. Finally, the two edges $e_{3,1}$ and $e_{3,2}$ encoding the truth value of the literal $l_3$ attach to $u_{3,c}$ and $u_{3,a}$, respectively. Hence, if literal $l_1$ is \texttt{true}, matrix $M_1$ must have a permutation of its columns such that column $a$ precedes column $b$, while if literal $l_1$ is \texttt{false}, matrix $M_1$ must have a permutation of its columns such that column $a$ follows column $b$. Analogously, the truth value of literal $l_2$ determines whether in matrix $M_2$ column $b$ precedes or follows column $c$, and the truth value of literal $l_3$ determines whether in matrix $M_3$ column $c$ precedes or follows column $a$.
	
	It follows that, if all three literals are \texttt{true}, then they induce unsatisfiable constraints on the ordering of the columns of the matrices, since column $a$ should precede column $b$, $b$ should precede column $c$, and $c$ should precede $a$. The same holds if all three literals are \texttt{false}. It can be easily checked that, for any other combination of truth values of the literals, there exists an ordering of the columns of matrices $M_1$, $M_2$, and $M_3$ that makes a planar drawing of the edges possible. Therefore, the constructed instance of \nt planarity with fixed sides admits a planar \nt representation if and only if the original instance of NAE3SAT admits a solution.\qed
\end{proof}

\section{The Free Sides Scenario}\label{se:free-sides}

In this section we extend the hardness result of Theorem~\ref{th:fixed-sides-completeness} to the free sides model and show that \nt planarity testing remains NP-complete when the maximum cluster dimension is larger than four. 

\begin{theorem}\label{th:free-sides-completeness}
\nt planarity testing with free sides and cluster size at most $k$ is NP-complete for any $k>4$.
\end{theorem}

	\begin{figure}[tb]
		\centering
		\includegraphics[width=7.2cm,page=2]{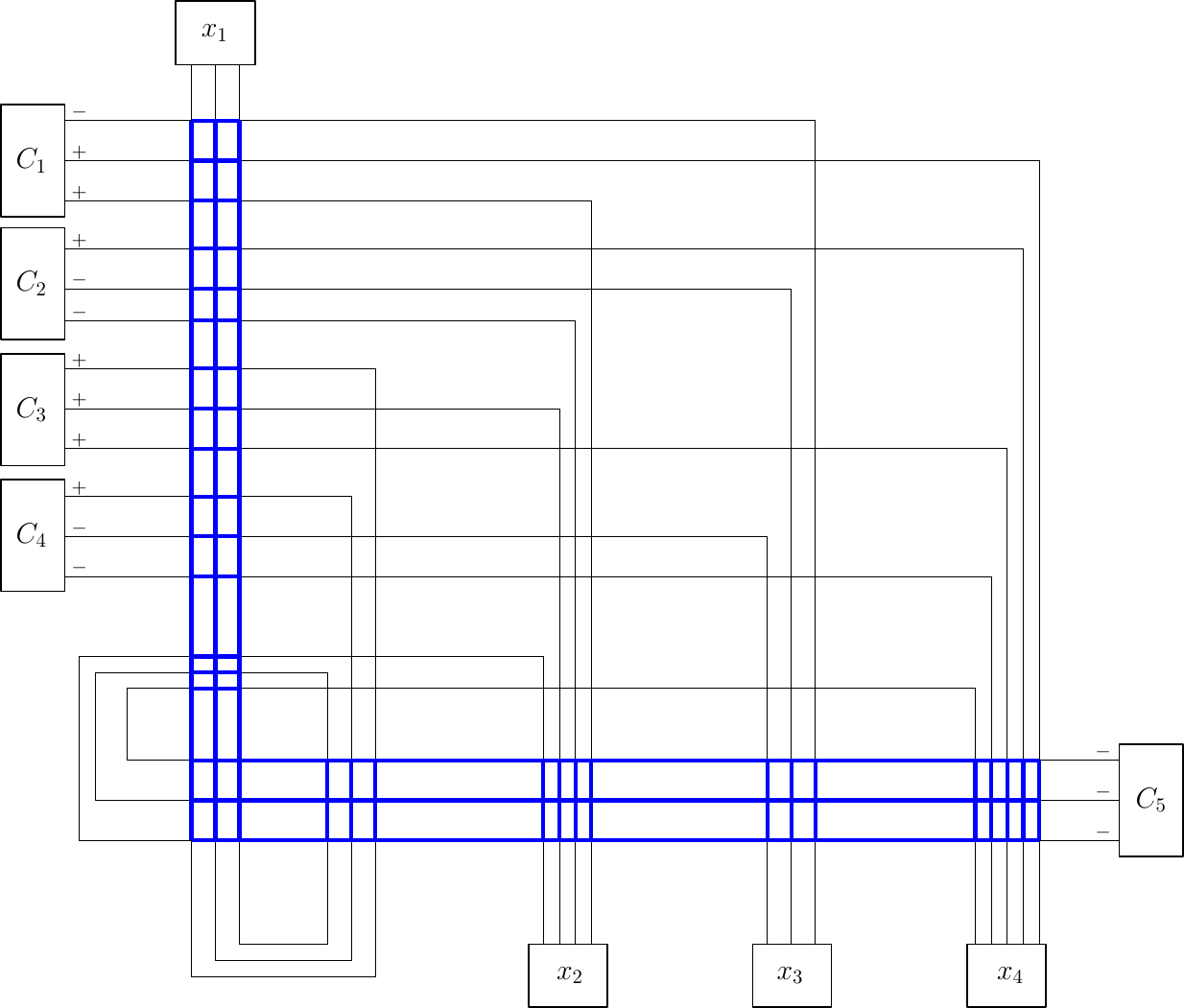}
		\caption{A non-planar drawing of an instance of NAE3SAT such that, by replacing crossings with dummy nodes, a planar and triconnected graph is obtained.}\label{fi:nae3sat-instance-3conn}
	\end{figure}

\begin{proof}
	The proof is based on a reduction from NAE3SAT. Starting from an instance $\varphi$ of NAE3SAT, we first build the instance $G_{\texttt{fix}}$ of \nt planarity with fixed sides by a construction that is similar to the one in the proof of Theorem~\ref{th:fixed-sides-completeness}.
	Namely, instead of starting from a drawing of $\varphi$ like the one in Fig.~\ref{fi:nae3sat-instance}, we start from a drawing $\Gamma$ like the one in Fig.~\ref{fi:nae3sat-instance-3conn}. The difference is that one variable $x_i$ (the leftmost one) and one clause $C_j$ (the bottommost one) are placed at the opposite side with respect to all the other variables and clauses, respectively. Moreover, their incident edges are not drawn as an L-shape, but as a polyline with three or five bends. More precisely, each edge that connects $x_i$ to a clause different from $C_j$ is drawn with three bends; each edge that connects $C_j$ to a variable different from $x_i$ is also drawn with three bends; the edge $(x_i,C_j)$, if it exists, is drawn with five bends. Clearly $\Gamma$ can be computed in polynomial time and has a polynomial number of crossings.	
	$G_{\texttt{fix}}$ is obtained by replacing variables, clauses, and crossings of $\Gamma$ with the same (variable, not, crossing, and clause) gadgets as in the proof of Theorem~\ref{th:fixed-sides-completeness}. 
	Clearly, $G_{\texttt{fix}}$, which is an instance of \nt planarity with fixed sides, admits a planar \nt representation if and only if the original instance $\varphi$ of NAE3SAT admits a solution. Furthermore, while the construction of Theorem~\ref{th:fixed-sides-completeness} may lead to separation pairs in the planarized drawing $\Gamma$ (shown with dashed arrows in the example of Fig.~\ref{fi:nae3sat-instance}) the modified technique described above guarantees that $\Gamma$ is triconnected. In fact, it consists of a subgraph of an orthogonal grid (shown with blue thick edges in Fig.~\ref{fi:nae3sat-instance-3conn}) together with edges that pairwise cross at most once.

\begin{figure}[tb]
	\centering
	\subfigure[]{\includegraphics[width=4cm]{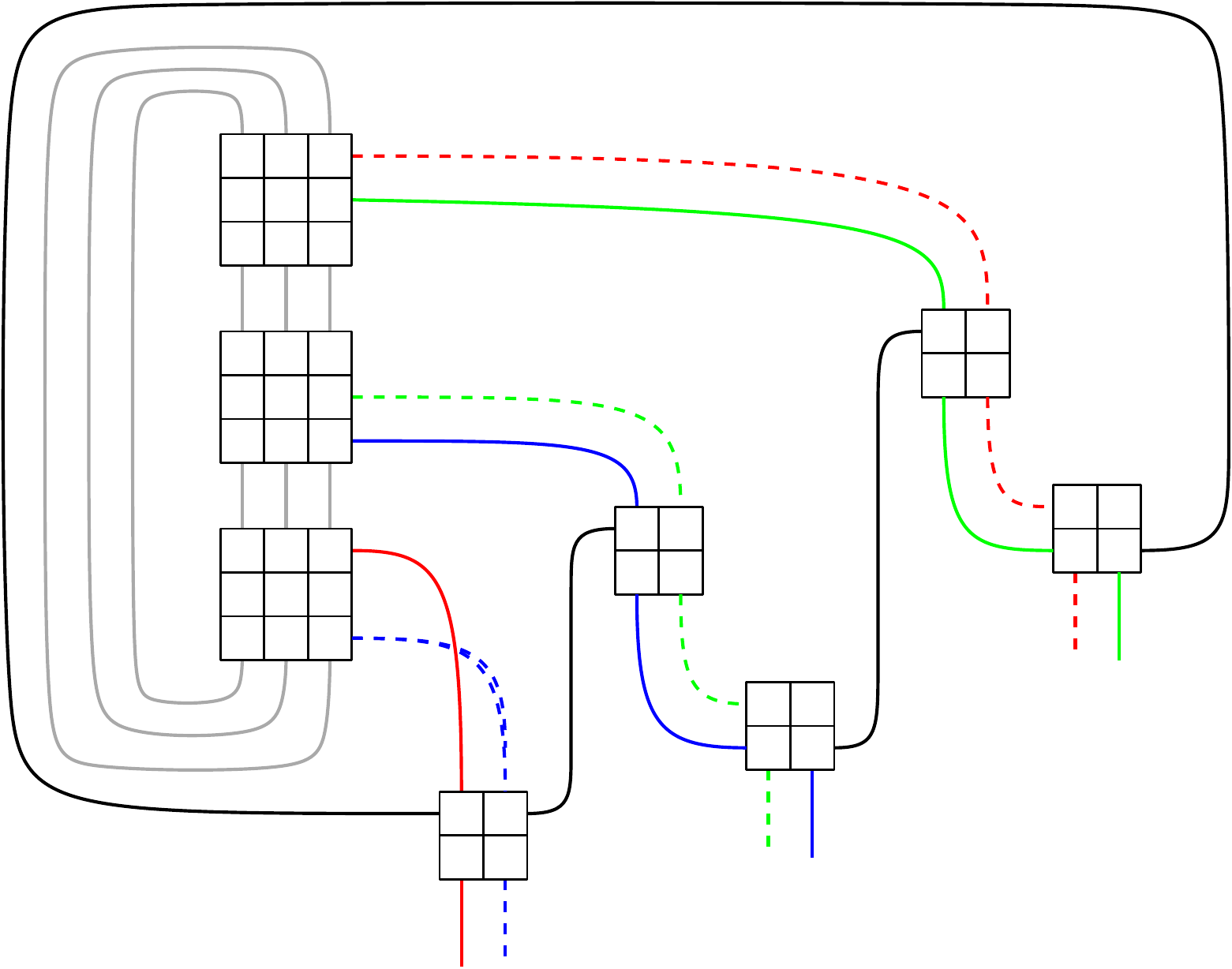}\label{fi:clause-gadget-3conn}}
	\hspace{0.5cm}
	\subfigure[]{\includegraphics[width=7cm]{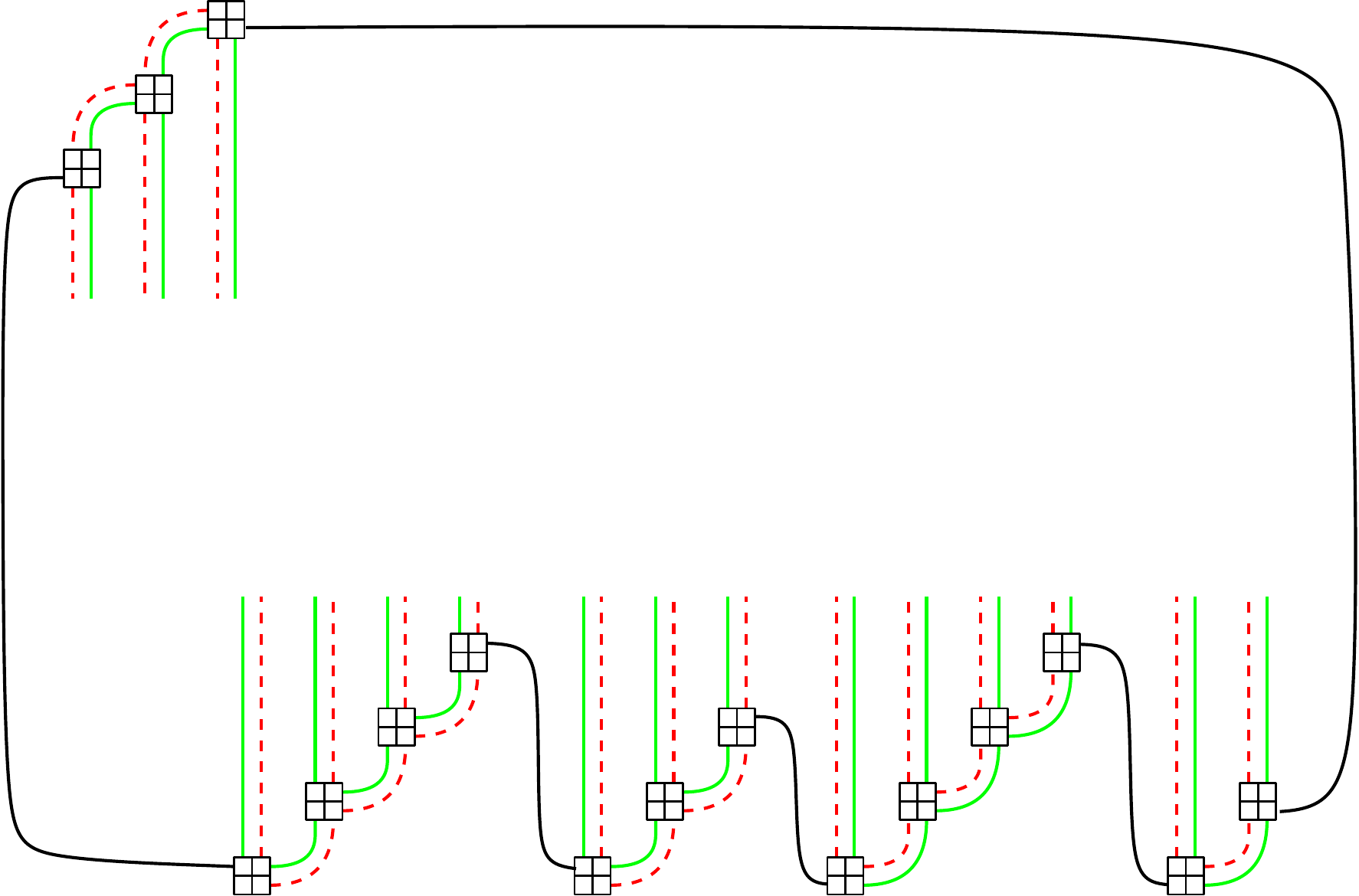}\label{fi:variable-3conn}}
	\caption{The extra-edges added to instance $G_{\texttt{fix}}$ to ensure that it has a triconnected frame-graph. \subref{fi:clause-gadget-3conn} A clause gadget with two negated literal and a directed one. \subref{fi:variable-3conn} The cycle connecting all variable gadgets.}\label{fi:3conn}
\end{figure}

    Observe that, although $\varphi$ is triconnected and the planarization of~$\Gamma$ yields a triconnected graph, the insertion of not gadgets and variable gadgets in  $G_{\texttt{fix}}$ introduces clusters that are adjacent to only two other clusters, producing a degree-two vertex in the frame graph of $G_{\texttt{fix}}$. Consider a clause gadget (refer to Fig.~\ref{fi:clause-gadget-3conn}). We add extra edges to connect the size-two clusters of the same clause gadget in a cycle as shown in Fig.~\ref{fi:clause-gadget-3conn}, in such a way that each size-two cluster has now at least three inter-cluster edges connecting it to other three distinct clusters.

    The same strategy is used for variable gadgets, where all the variable gadgets can 
    be linked together in a cycle enclosing the whole instance (see Fig.~\ref{fi:variable-3conn}). Again, after the addition of extra edges each size-two cluster of the variable gadgets has at least three inter-cluster edges connecting it to other three distinct clusters. It can be checked that, after the above described changes to $G_{\texttt{fix}}$, the frame graph of $G_{\texttt{fix}}$ has no separation pair, i.e., it is triconnected.
         
	The modified construction of $\Gamma$ used to construct $G_{\texttt{fix}}$ guarantees that the frame graph of $G_{\texttt{fix}}$ is triconnected. This is a consequence of the fact that $\varphi$ is triconnected and that inserting crossing gadgets is equivalent to planarizing $\Gamma$.

	\begin{figure}[tb]
		\centering
		\subfigure[]{\includegraphics[width=3.2cm]{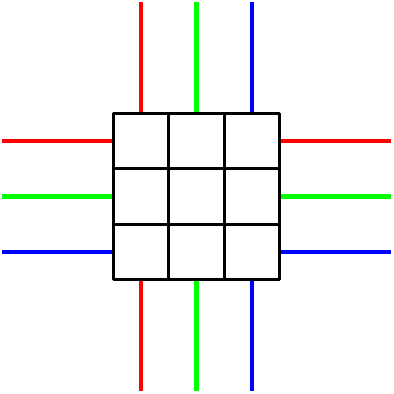}\label{fi:matrix-replacement-3x3}}
		\hspace{0.5cm}
		\subfigure[]{\includegraphics[width=3.2cm,page=1]{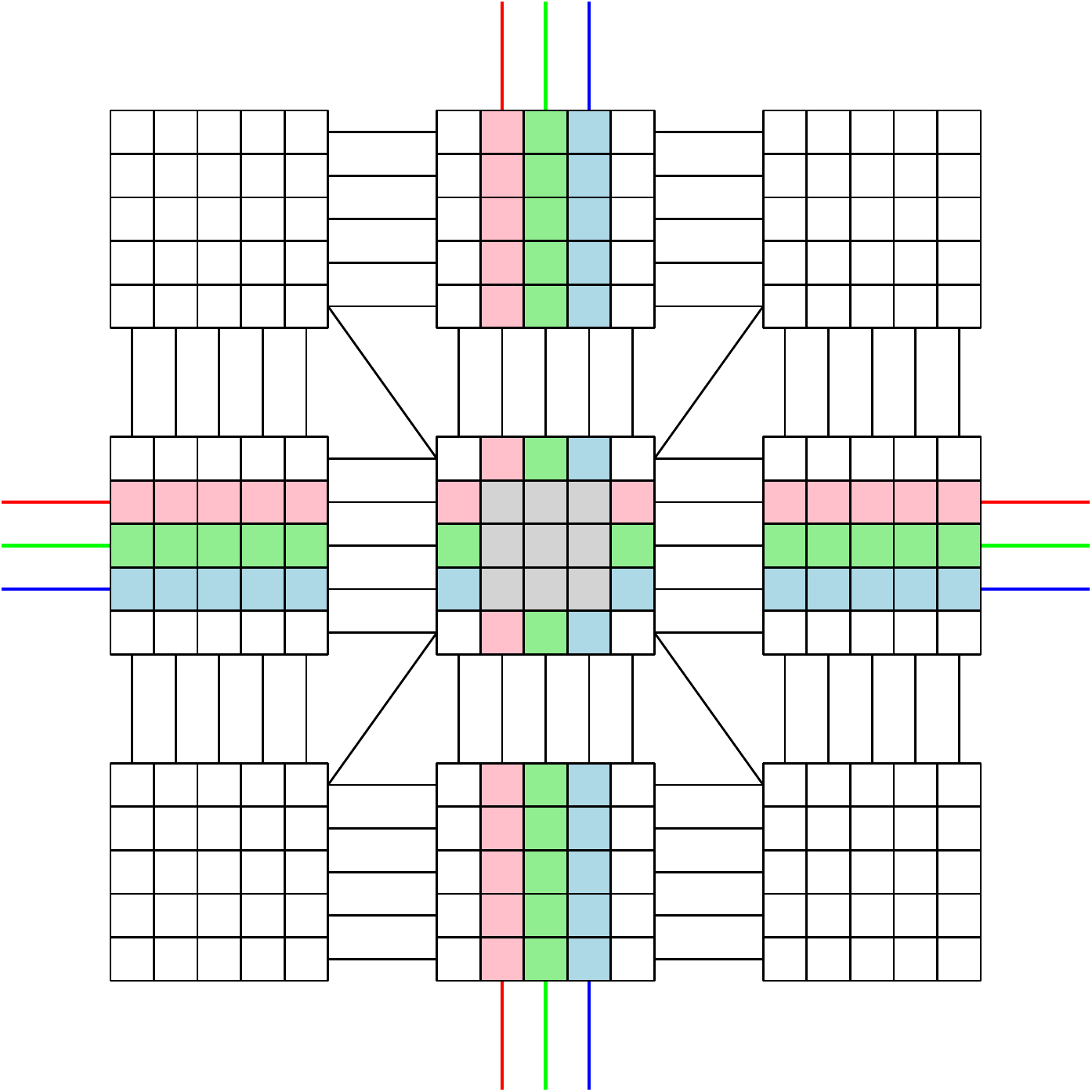}\label{fi:matrix-replacement-5x5}}
		\hspace{0.5cm}
		\subfigure[]{\includegraphics[width=3.2cm,page=2]{matrix-replacement-5x5.pdf}\label{fi:matrix-replacement-5x5-frame}}
		\caption{\subref{fi:matrix-replacement-3x3} A cluster $V_i$ of size three of the instance of \nt planarity with fixed sides. \subref{fi:matrix-replacement-5x5} The gadget used to replace $V_i$ in the instance of \nt planarity with free sides.
		\subref{fi:matrix-replacement-5x5-frame} The wheel graph of the frame graph of $G_{\texttt{free}}$ corresponding to $V_i$.}\label{fi:matrix-replacement}
	\end{figure}
	
	We now construct an instance $G_{\texttt{free}}$ of \nt planarity with free sides by replacing each cluster of maximum size three of $G_{\texttt{fix}}$ with the gadget depicted in Fig.~\ref{fi:matrix-replacement-5x5} that uses exclusively clusters of size $5$. The gadget consists of nine clusters such that the corresponding nodes in the frame graph of $G_{\texttt{free}}$ form a wheel graph with an external cycle of eight nodes all connected to a central one (see Fig.~\ref{fi:matrix-replacement-5x5-frame}). Since the wheel has only one embedding (up to a flip), the gadget admits a \nt planar representation only if the hub of the wheel is drawn inside the cycle formed by the other eight clusters. Also, the edges that in the instance $G_{\texttt{fix}}$ are constrained to attach to a specific side of a matrix due to the side assignment $\Phi$, are now all incident to the same cluster of the wheel. In other words, although the free model allows to attach these edges to any side of the matrix they are incident to, the gadget forces these edges to attach to a specific side of the gadget itself. Since $G_{\texttt{free}}$ has a triconnected frame graph, the embedding of the frame graph is fixed, and a \nt planar representation with fixed sides of $G_{\texttt{fix}}$ exists if and only if a \nt planar representation with free sides of $G_{\texttt{free}}$ exists.
	\qed
\end{proof}

\section{Conclusions and Open Problems}\label{se:open-problems}

This paper concerns \nt planarity testing of flat clustered graphs with clusters of size at most $k$, where $k$ is a given parameter. The main focus is on the fixed sides scenario, which assumes  the side of the matrix to which an edge is incident to be fixed as a part of the input. In this scenario it is proved that \nt planarity testing can be solved in linear time for $k=2$, but it becomes NP-complete for $k>2$. The reasons of the problem complexity reside in the rigid components of the SPQR-tree that decomposes the graph obtained by collapsing its clusters into vertices. Namely, when the decomposition tree does not have R-nodes we prove that NodeTrix planarity testing of flat clustered graphs is solvable in linear time for any fixed  $k>2$. The hidden multiplicative factor in the linear time complexity of our algorithm is exponential in $k$.

In addition to the above results, we study NodeTrix planarity testing in the free sides scenario, i.e. assuming that  the side of the matrix to which an edge is incident is not given as a part of the input. We show that NodeTrix planarity testing in the free sides scenario remains NP-complete for values of  $k$ such that  $k>4$.

We conclude the paper by mentioning some open problems that are naturally suggested by the research in this paper.

\begin{enumerate}
	
	\item Study the complexity of the NodeTrix planarity testing problem in the free sides scenario for values of $k$  such that $2 \leq k \leq 4$.
	
	\item For values of $k > 4$, study families of clustered graphs for which NodeTrix planarity testing is fixed parameter tractable in the free sides scenario.
	
	\item Study whether other types of hard hybrid planarity testing problems such as, for example, planar intersection-link representability~\cite{addfpr-ilrg-17}, admit a polynomial time solution when the size of the clusters is bounded by a constant.
	
    \item What is the complexity of NodeTrix planarity testing when edges cannot be attached to the top side of the matrices? This would allow, for example, to equip the matrices with labels and it is, therefore, a natural restriction of the model addressed in this paper. 
    
    \item In our definition we require that the clusters are represented by symmetric adjacency matrices. This is consistent with all previous papers that use NodeTrix representations. One could however study a different model in which rows and columns can be permuted independently.     

\end{enumerate}

\bibliography{biblio}
\bibliographystyle{splncs04}

\end{document}